%% file: main.tex
\documentclass[acmsmall,nonacm]{acmart}
\AtBeginDocument{%
  \providecommand\BibTeX{{%
    \normalfont B\kern-0.5em{\scshape i\kern-0.25em b}\kern-0.8em\TeX}}}

\usepackage{enumitem}
\usepackage{listingsutf8}
\usepackage{mathpartir}
\usepackage{pf2}
\usepackage{xparse}
\usepackage{xspace}

\usepackage{contour}
\usepackage{ulem}

\contourlength{0.8pt}

\newcommand{\myuline}[1]{%
  \uline{\phantom{#1}}%
  \llap{\contour{white}{#1}}%
}

\usepackage{macros}
\usepackage{while_macros}
\usepackage{js_macros}
\usepackage{c_macros}
\usepackage{gil_listings}

\begin{document}

\title[\gillian: A Multi-Language Platform for Unified Symbolic Analysis]{\gillian: A Multi-Language Platform for \\ Unified Symbolic Analysis}

\author{Petar Maksimovi\'c}
\affiliation{
  \institution{Imperial College London}
  \country{UK}                     
}
\email{p.maksimovic@imperial.ac.uk}         

\author{Jos\'e Fragoso Santos}
\affiliation{
  \institution{INESC-ID/Instituto Superior T\'ecnico, Universidade de  Lisboa}
  \country{Portugal}                   
}
\email{jose.fragoso@tecnico.ulisboa.pt}          

\author{Sacha-\'Elie Ayoun}
\affiliation{
  \institution{Imperial College London}          
  \country{UK}                    
}
\email{s.ayoun@imperial.ac.uk}         

\author{Philippa Gardner}
\affiliation{
  \institution{Imperial College London}              
  \country{UK}          
}
\email{p.gardner@imperial.ac.uk} 

\begin{abstract}
This is an evolving document describing the meta-theory, the implementation, and the instantiations of \gillian, a multi-language symbolic analysis platform.
\end{abstract}

\maketitle

\input{sections/introduction}

\input{sections/mms}

\input{sections/sms}
\input{sections/gil}
\input{sections/spec}

\bibliographystyle{ACM-Reference-Format}
\bibliography{main}

\end{document}
\endinput

%% file: sections/introduction.tex
\section{Introduction}
\label{sec:intro}

\gillian was introduced in~\cite{gillianpldi} as a multi-language platform for whole-program symbolic execution, parametric on the concrete and symbolic memory models of the target language (TL), and underpinned by a core symbolic execution engine with strong mathematical foundations.
Gillian analysis is done on GIL, an intermediate goto language
parametric on a set of \emph{memory actions},  which describe the fundamental ways in which TL programs interact with their memories.
To instantiate Gillian to a new TL, a tool developer must: (1)
identify the set of the TL memory actions and implement the TL memory models 
using these  actions; and (2) provide a trusted compiler from the TL
to GIL, which preserves the TL memory models and the semantics. 
In~\cite{gillianpldi}, \gillian was instantiated to JavaScript (JS)
and C, and these instantiations, called \gillian-JS and \gillian-C, were used to find bugs in two real-world data-structure~libraries.

In~\cite{gilliancav}, \gillian was extended with support for compositional memory models and verification based on separation logic.
The compositional memory models of \gillian  work with \emph{partial} memories and are formulated in terms of \emph{core predicates} and associated \emph{consumer} and \emph{producer} actions for the
TL memory models, which need to be provided by the tool developer. The core predicates describe the fundamental units of
TL memories: e.g., a JS object-property pair and a C block
cell. The consumers and producers, respectively, frame off and frame on the TL memory resource described by the core predicate. The partial memories also need to track {\em negative} resource: that is, the resource known to be absent from the partial memory.
Gillian verification is built on top of compositional memory models. In particular, the core predicates induce an assertion language
for writing function specifications in separation logic and the 
consumers and producers allow for the creation of
 a fully parametric spatial entailment
engine, enabling re-use of function  specifications in symbolic execution.
\gillian also allows tool developers to extend assertions
with 
user-defined predicates
so as to identify the TL language
interface familiar to code developers, and code developers to provide additional
predicates and lemmas to verify the particular data
structures in their~programs.
In~\cite{gilliancav}, Gillian-JS and Gillian-C were extended to support verification, and used to provide verified specifications of the JS and C implementations of the deserialisation module of the AWS Encryption~SDK, discovering two bugs in the former and three in the latter.

\myparagraph{Outline}
This document currently contains the following content:
\begin{description}[labelindent=1em]
\item[(\S\ref{sec:exec})] an account of \gillian's whole-program  execution, including:
\begin{description}[labelindent=0.6em,labelwidth=2.5em,leftmargin =!]
\item[(\S\ref{subsec:gilsyn})] the syntax of \gillian's intermediate language, \gil
\item[(\S\ref{subsec:mem})] memory models for whole-program execution and their properties, which require the \emph{memory actions} of the target language
\item[(\S\ref{sub:alloc})] allocators, which relieve \gillian users of the need to reason explicitly about allocation
\item[(\S\ref{subsec:esm})] state models, which are built on top of memory models, and their properties
\item[(\S\ref{sub:execsem})] the single-trace and the collecting \gil semantics, defined in terms of state models
\item[(\S\ref{sec:cor:tbf})] the correctness result stating that the \gil symbolic execution does not have false positive bug reports
\item[(\S\ref{sec:cor:bc})] partial verification results for the \gil symbolic execution.
\end{description}
\item[(\S\ref{sec:gilspec})] an account of \gillian's compositional execution, including:
\begin{description}[labelindent=0.6em,labelwidth=2.5em,leftmargin =!]
\item[(\S\ref{sec:cmms})] the compositional memory models for \gillian, defined in terms of \emph{core predicates} and their \emph{consumers} and \emph{producers}, together with their key properties
\item[(\S\ref{subsec:csms})] the compositional state models for \gillian, built on top of compositional memory models
\item[(\S\ref{sec:sublift})] a lifting of the correctness results given in \S\ref{sub:execsem} and \S\ref{sec:cor:bc} to compositional memory and state models
\item[(\S\ref{cor:fp})] the correctness result stating that the \gil compositional execution is frame-preserving up to a renaming of allocated symbols
\end{description}
\item[(\S\ref{subsec:gilasrt})] \gillian assertions and specification, including:
\begin{description}[labelindent=0.1em,labelwidth=2.4em,leftmargin =!]
\item[(\S\ref{subsec:memasrt})] \gillian memory assertions, their consumers and producers, and their associated properties
\item[(\S\ref{subsec:stasrt})] \gillian state assertions, built on top of memory assertions, their consumers and producers, and their associated properties
\item[(\S\ref{subsec:asrtsat})] the parametric assetion satisfiability relation and its connection to consumers and producers
\item[(\S\ref{subsec:procspecs})] \gillian procedure specifications for verification, together with an extension of the \gil semantics that allows for use of specifications in \gil execution, the corresponding verification correctness result, and a description of the \gillian verification procedure.
\end{description}
\end{description}

%% file: sections/mms.tex

\newcommand{\mems}{\mathsf{M}}
\newcommand{\smems}{\hat{\mems}}

\section{\gillian Whole-Program Execution}
\label{sec:exec}

We present  \gil, which is parametric on the memory models of  the
target language (TL). 
We first  define arbitrary memory models, which we then 
instantiate to arbitrary concrete and symbolic memory models.
We lift arbitrary the memory models to arbitrary state models, and
define the \gil execution semantics  on states. 
We finally present requirements that the TL memory models must
satisfy in order for the resulting \gil symbolic analysis to be
correct,  establishing \emph{true bug-finding} and
\emph{bounded correctness} for  the \gillian execution semantics. 

\subsection{\gil Syntax}
\label{subsec:gilsyn}

\gil is parametric on a set of \myuline{memory actions} of the target
language (TL), $\mactions \ni \act$, which capture the fundamental
ways in which TL programs interact with their memories. 
Its syntax is given~below. 

\smallskip
\begin{display}{\gil Syntax}\label{app:gil:syntax}
\begin{tabular}{l@{}l}
   $\vl \in \vals \defeq i, j, \num \in \nums \mid \str \in \strs \mid
  \bool \in \bools \mid \loc, \varsigma \in \locs \mid \type \in
  \types \mid \fid \in \fids \mid \lst{\vl} \in List(\vals)$ &
   \\[2pt]
   $\e \in \exprs \defeq \vl \mid \x \in \xs \mid \unop{\e} \mid \binop{\e_1}{\e_2} \hspace*{1.8cm} \sexp \in \sexps \defeq \vl \mid \lx \in \lxs \mid \unop{\sexp} \mid \binop{\sexp_1}{\sexp_2}$
   \\[2pt]
    $\cm \in \cmds{}$ $\defeq$ $\x := \e \mid \ifgoto{\e}{i} \mid \x := \e(\e') 
                                     \mid \x := \act(\e) \mid$ & \hspace*{-2.1cm}   $\proc \in \procs{} \defeq \procedure{f}{\x}{\lst{\cm}}$
\\  
    \hspace*{1.5cm}$\x := \symb{}(\e) \mid \x := \fresh{}(\e) \mid \return{\e} \mid  \fail{e} \mid \vanish$  & \hspace*{-2.1cm}  $\prog \in \progs{} :  \fids \pmap \procs{}$
 \end{tabular}
\end{display}

\smallskip
The set of \gil values, $\vl \in \vals$, contains numbers, strings, booleans,
uninterpreted symbols\footnote{Uninterpreted symbols can be used to
  represent, for example, memory locations, $\loc \in \mathcal{L} \subset \locs$,  or TL-specific constants. },
simple types (e.g. types for  numbers, strings, booleans and lists), function identifiers, and lists of values. 
The set of \gil expressions, $\e \in \exprs$, contains values, program
variables, and various unary and binary operators such as  those  used
for number, string, boolean, and list manipulation. The set of \gil symbolic
expressions, $\sexp \in \sexps$,  is  constructed in the same way as
for 
\gil expressions, except that it is built from symbolic variables, $\lx
\in \lxs$, instead of program variables, and we assume it contains
the  standard set of symbolic 
Boolean expressions. 

\gil commands, $\cm \in \cmds{}$, contain variable assignment,
conditional goto, procedure call, memory actions, commands for allocation of uninterpreted and
interpreted symbols, function return, error termination and path cutting. Their semantics will be explained in more detail in \S\ref{sub:execsem}.
A \gil procedure, $\procedure{f}{\x}{\lst{\cm}}$, comprises an
identifier $f \in \fids $, a formal
parameter~$x$\footnote{Without loss of generality, we describe one-parameter procedures for
  convenience in the \gillian theory, using multiple parameters in the
  \gillian 
  implementation.},  and a body
given by a list of commands $\lst{\cm}$.
A \gil program is a set of \gil procedures. 

\subsection{Execution Memory Models}
\label{subsec:mem}

The execution memory models defined below are designed for whole-program execution. They differ from those introduced in~\cite{gillianpldi} in that they unify the concrete and symbolic memory models of~\cite{gillianpldi} and also introduce memory well-formedness.
An execution memory model is parametric on a value set, $\gval \in \gv$, and
an action  set, $\actions \in \alpha$. The value set must contain the \gil values: that is, 
$\vals \subseteq \gval$. We distinguish the Boolean value set, $\pcs  \subset \gval$, and refer to $\pc \in \pcs$ as a \emph{context}. For example, for concrete execution, the contexts are the Booleans $\true$ and $\false$ whereas, for symbolic execution, they are Boolean symbolic expressions such as $\lx \land \hat y$.

\begin{definition}[Execution Memory Model]\label{def:emm}
Given a value set $\gval$ and an action set  $\actions$, an
\myuline{execution}
\myuline{memory model},
$\cmemory(\gval, \actions) \defeq \tup{\sset{\cmemory}, \sinv, \cea}$, 
comprises:
\begin{enumerate}[leftmargin=*,label=\arabic*)]
\item a set of memories $\sset{\cmemory} \ni \cmem$; 
  \item a \myuline{well-formedness relation}, $\sinv \subseteq
    \sset{\cmemory} \times \pcs$, with  $\sinv_{\!\pc}(\cmem)$ denoting 
    that memory~$\cmem$ is well-formed in context $\pc$, which has the following properties:
    \begin{enumerate}[left=0cm,label=\arabic{enumi}.\arabic*)]
    \item $\pi$ is satisfiable\footnote{We observe that Boolean-typed values can be lifted to meta-logic propositions and, for simplicity, overload the $\pc/\pcs$ notation to denote both Boolean-typed values of the memory model and meta-logic propositions, eliding the conversions.}, written $\pi \;\sat$; 
      \item well-formedness is monotonic with respect to context
  strengthening:  $$\sinv_{\!\pc}(\cmem) \land (\pc' \Rightarrow \pc)
  \land \pc'~\sat \implies \sinv_{\!\pc'}(\cmem);$$
  \end{enumerate}

\item 
an \myuline{action execution function}, $\cea : \sset{\cmemory}
  \times  \actions \times \gval \times \pcs \pmap \power{\sset{\cmemory} \times
    \gval \times \results \times \pcs }$\footnote{The action result, $\result \in \results =\{ \osucc, \oerr, \omiss \}$, denotes success ($\osucc$), a non-correctible error ($\oerr$), or a missing resource error ($\omiss $).},
  where $\cea  (\mu, \alpha, \gv, \pc) =  \{ (\cmem_i, \gv_i , \result_i
  ,\pc_i)     : i  \in I \}$, for countable $I$, with notation 
  $\fcolactionx \cmem \act \gv {\pc} {\left\{ (\cmem_i,
      \gv_i)_{\pc_i}^{\result_i}|_{i \in I} \right\}}$ for all
  outcomes and  $\faction{\cmem}{\act}{\gv}{\pc}{\cmem_i}{\gv_i}{\pc_i }
{\result_i}$ for  a specific outcome, 
 satisfies the following properties:
 \begin{enumerate}[left=0cm,label=\arabic{enumi}.\arabic*)]
 \item $\sinv_{\!\pc}(\cmem) $
   \item action execution preserves memory well-formedness: $
\sinv_{\!\pc_i}(\cmem_i)
$ 
  \item $\pi_i~\sat \land \pi_i \Rightarrow \pi$, for all $ i \in I$;
  \item the $\pi_i$ are separately disjoint: that is, 
    $\lnot(\pc_i \land \pi_j)$,  for all $ i, j \in I, i\neq j$;
  \item the $\pi_i$ have full coverage: that is, 
    $\pc \implies \bigvee_{i \in I} \pc_i$;
\item non-successful action execution does not affect the memory:
  $  r_i \neq \osucc 
  \implies \cmem_i  \equiv  \cmem$
  \end{enumerate}
 \end{enumerate}
\end{definition}

\newcommand{\ctxt}{\mathcal{C}\mkern-0.5mutx}
\newcommand{\ctx}{\mathit{ctx}}

\begin{definition}[Concrete and Symbolic Execution Memory Model]\label{def:csemm}
A {\em concrete} execution memory model, 
$\cmemory(\vals, \actions)$, is an execution memory model with 
value set given by the \gil value set, $\vals$. 
A {\em symbolic} execution memory model, $\cmemory(\sexps, \actions)
$,  is an execution memory model with value set given by the  set of \gil symbolic expressions, $\sexps$. Onward, we denote concrete memory models by $\cmemory(\actions)$, and symbolic state models by $\smemory(\actions)$.
\end{definition}

\myparagraph{Discussion and Observations} Properties 3.1 and 3.2 connect action execution to well-formedness of memories in an expected way. Property 3.3 states that memory action execution cannot lose information; this captures the well-known property of symbolic execution, where path conditions can only get stronger. 
Property 3.4 disallows non-determinism that is not captured by the context. This, given that the only satisfiable context for concrete memory models is the \gil
Boolean $\true$, means that
the concrete memory action execution does not branch
and contexts can be elided for concrete well-formedness and action
execution, denoted by $\sinv(\cmem)$ and
$\fcaction{\cmem}{\act}{\vl}{}{\cmem'}{\vl'}{} \result$ respectively. In fact, the following general property, which captures action determinism, holds: 
$$\faction{\cmem}{\act}{\gv}{\pc}{\cmem'}{\gv'}{\pc}
{\result} \implies \fcolactionx{\cmem}{\act}{\gv}{\pc}{\{ (\cmem', \gv')^\result_\pc}\}$$
Property 3.5 effectively means that memory action execution does not drop paths, and is essential for exact and over-approximate reasoning. Property 3.6 allows us to do easier error correction.

.

%% file: sections/sms.tex
\subsection{\gillian Allocators}
\label{sub:alloc}

\gillian comes with built-in allocators that allow tool developers to
create fresh memory locations, symbolic variables, and values,
relieving them of the need to explicitly handle  allocation.

\begin{definition}[Allocator]
\label{def:am}
Given a value set $\gval$, an allocator
$\allocator(\gval) = \tup{\sset{\allocator}, \allocf}$
comprises:
\etag{1}~a set $\sset{\allocator} \ni \arec$ of allocation records, $\arec: \power{\mathsf{\gval}} \pmap  \power{\mathsf{\gval}}$, which partially map allocation ranges (subsets of $\gval$) to sets of values allocated from that range;
\etag{2}~and 
an allocation function:
$$\allocf : \sset{\allocator} \tmap \mathbb{N} \tmap \power{\gval} \pmap \sset{\allocator} \times \gval~\mathtt{list},
$$
pretty-printed as
$\alloc{\arec}{k}{\arec', {y}}{Y}$,
which takes an allocation record $\arec$, the number of fresh values to allocate, $k$, and an allocation range $Y \subseteq \gval$, and returns $k$ fresh values, ${y_1, \ldots, y_k} \in Y$, together with the updated allocation record,~$\arec'$. Let $\arec_o(Y)$ denote $\arec(Y)$ if $Y \in \domain(\arec)$, and $\emptyset$ otherwise. Then, the allocation function is defined as follows:
\begin{mathpar}
\small
\inferrule[Allocation]
{
\mathsf{fresh}_Y(\arec_o(Y), k) \rightsquigarrow \lst{y} \and \arec' = \arec[Y \mapsto \arec_o(Y) \cup \lst{y}]}
{\alloc{\arec}{k}{\arec', \lst{y}}{Y}}
\end{mathpar}
where $\mathsf{fresh}_Y(\lst{y_{\mathit{a}}}, k) \rightsquigarrow \lst{y_{\mathit{f}}}$ means that the $k$ values in $\lst{y_{\mathit{f}}}$ are in some way considered \emph{fresh} w.r.t.~the values in $\lst{y_{\mathit{a}}}$, with all values belonging to the set $Y$.
\end{definition}

As for the memory models, a {\em concrete} allocator is defined using the value set $\vals$, and a
{\em symbolic} allocator is defined using the value set $\sexps$. We define the $\mathsf{fresh}_Y(\lst{y_f}, \lst{y_a})$ function for $Y \in \{ \clocs, \lxs, \vals \}$, because these are the use cases in the concrete and symbolic execution, where memory locations and symbolic variables are allocated with the usual notion of freshness, and for \gil values, any value is allowed to be allocated any number of times.
\begin{mathpar}
\small
\inferrule[Fresh ($\clocs$ and $\lxs$)]
{Y \in \{\clocs, \lxs\} \and \lst{y_f} = [y_1, \ldots, y_k] \and 
y_i \in Y \and (y_i \notin Y_a)
 }
{\mathsf{fresh}_\clocs(Y_a, k) \rightsquigarrow \lst {y_f}}
\and
\inferrule[Fresh ($\vals$)]
{\lst{v_f} = [v_1, \ldots, v_k] \and 
v_i \in \vals 
 }
{\mathsf{fresh}_\clocs(-, k) \rightsquigarrow \lst {v_f}}
\end{mathpar}

\subsection{Execution State Models}
\label{subsec:esm}

We define the arbitrary \gillian execution state models,
used to define 
the \gil execution semantics, and 
describe how
lift arbitrary execution memory models to execution state
models. 
Assume a set of \myuline{basic state actions},  $\sactions
= \{ \kwT{setVar}_{\x} \}_{\x \in \xs} \cup \set{\kwT{setStore}, \kwT{getStore}} 
\cup \{ \kwT{eval}_{\e} \}_{\e \in \exprs} \cup \{\kwT{assume},$ $\kwT{uSym}, \kwT{iSym}\}
$, which are used in the definition of an execution state model to capture
variable store manipulation, expression evaluation, branching
and symbol allocation\footnote{We assume that the basic state actions
  always succeed, and that \gil expressions are not
  side-effecting. Errors due to non-existent variables in the store or
  mistyped expressions could be modelled analogously to the memory
  action errors, but we choose not to do this. Side-effecting expressions of the TL can always be compiled to non-side-effecting expressions.}.

\begin{definition}[Execution State Model]
\label{def:esm}

  Given a value set $\gval$, a set of interpreted symbols $\mathsf{I}
\subset \gval$,
an action set $A$,  an execution memory model $\cmemory(\gval, A)
\defeq \tup{\sset{\cmemory}, \sinv,$  $\cea}$, and an allocator
$\allocator(\gval) = \tup{\sset{\allocator}, \allocf}$,
an {\em execution state model} is defined by 
$\gstate(\gval,
\mathsf{I}, A, \cmemory(\gval, A), \allocator(\gval)) \defeq
\langle\sset{\gstate}, \ceval{~}{}, \kwT{ea} \rangle$, and it consists of:
\begin{itemize}[leftmargin=*]
\item a set of \myuline{states}, $\sset{\gstate} \subseteq 
  \sset{\cmemory} \times (\xs \pmap \gval) \times \sset\allocator
  \times \pcs$,  with  $\st = \tup{\cmem,
    \sto, \arec, \pc} \in \sset{\gstate}  $ { if and only if $\sinv_{\!\pc}(\cmem)$};
  
\item an {\em expression evaluation function} $\ceval{~}{} : \exprs
  \times (\xs \pmap \gval) \tmap \gval$, where $\ceval{\e}{\sto} \in
  \gval$ denotes the evaluation of a \gil expression $e$ with respect
  to variable store $\sto$;

\item a \myuline{state action execution function}, $\kwT{ea} :
  \sset{\gstate} \times (\mactions \uplus \sactions) \times \gval
  \pmap \power{\sset{\gstate} \times \gval \times \results}$,
  defined by:
  for $\st = \tup{\cmem,
    \sto, \arec, \pc}$ and $\alpha \in A$,
  
\medskip
\noindent
\begin{tabular}{l@{~}c@{~}l}
     $\kwT{ea}(\st, \act, \gv)$
    & $\semeq$ & 
    $\{ (\tup{\cmem', \sto, \arec, \pc'}, \gv')^\result \mid  \faction \cmem {\act} \gv {\pc} {\cmem'} {\gv'} {\pc'} \result 
\} $ 
\\
  $\kwT{ea}(\st, \kwT{setVar}_{\x},  \gv)$ 
       & $\semeq$ & 
       $\set{ (\tup{\cmem, \sto[\x \mapsto \gv], \arec, \pc}, \true)^\osucc}$ \\ 
  $\kwT{ea}(\st, \kwT{setStore}, \sto')$
     & $\semeq$ & 
     $\set{(\tup{\cmem, \sto', \arec, \pc}, \true)^\osucc}$ \\
  $\kwT{ea}(\st, \kwT{getStore}, -)$
    & $\semeq$ & 
     $\set{ (\st, \sto)^\osucc}$ \\
  $\kwT{ea}(\st, \kwT{eval}_{\e}, -)$
    & $\semeq$ & 
     $\set{ (\st, \ceval{\e}{\sto})^\osucc}$ \\
  $\kwT{ea}(\st, \kwT{assume}, \pc')$
    & $\semeq$ & 
    $\lbrace (\tup{\cmem, \sto, \arec, \pc \land \pc'}, \true)^\osucc \mid \pc \land \pc'~\sat) \rbrace$ \\
    $\kwT{ea}(\st, \kwT{uSym}, n)$
    & $\semeq$ &
     $\lbrace (\tup{\cmem, \sto', \arec', \pc}, \lst \loc)^\osucc \mid \alloc{\arec}{n}{\arec', \lst \loc}{\clocs} \rbrace $ \\
    $\kwT{ea}(\st, \kwT{iSym}, n)$
    & $\semeq$ &
     $\lbrace (\tup{\cmem, \sto', \arec', \pc}, \lst \gv)^\osucc \mid
                 \alloc{\arec}{n}{\arec', \lst \gv}{\mathsf{I}}
                 \rbrace $ 
              
\end{tabular}
\end{itemize}
\end{definition}

\noindent {\em Notation.} Given $\kwT{ea}   (\sigma, \alpha, \gv) =  \{ (\sigma_i, \gv_i ,
\result_i)     | i  \in I \}$  for countable $I$,  we write 
 $\fcolactionx \sigma \act \gv {} {\left\{ (\sigma_i,
      \gv_i)^{\result_i}|_{i \in I} \right\}}$ for all
  outcomes and  $\faction{\sigma}{\act}{\gv}{}{\sigma_i}{\gv_i}{}
  {\result_i}$ for  a specific outcome. 
  Given a state $\st$, we access its components as follows: $\access{\st}{mem}$ for the memory, $\access{\st}{sto}$ for the store, $\access{\st}{al}$ for the allocator, and $\getpc{\st}$ for the context.

\begin{proposition}[State Properties]\label{prop:esm}
  The state action execution function of an execution state model, $\gstate(\gval,
\mathsf{I}, A, \cmemory(\gval, A), \allocator(\gval)) \defeq 
\langle\sset{\gstate}, \ceval{~}{}, \kwT{ea} \rangle$,
satisfies the following properties,
  for
 $\fcolactionx \sigma \act \gv {} {\left\{ (\sigma_i,
     \gv_i)^{\result_i}|_{i \in I} \right\}}$:
 
  \begin{enumerate}[leftmargin=*,label=\arabic*)]
  \item $\getpc{\st_i} \Rightarrow \getpc{\st}$, for all $i \in I$;
  \item contexts of final states are separately disjoint: i.e., 
    $\lnot(\getpc{\st_i} \land \getpc{\st_j})$,  for all $ i, j \in I, i\neq j$;
   \item all state actions except $\kwT{assume}$\footnote{This is not an issue as the coverage property will be re-established for the \gil execution (cf.~\S\ref{sub:execsem})} have full coverage: i.e., 
    $\getpc{\st} \implies \bigvee_{i \in I}\getpc{\st_i}$;
  \item unsuccessful action execution may only affect the context.
  \end{enumerate}

\end{proposition}

\begin{definition}[Concrete and Symbolic Execution State
  Model]\label{def:csesm}
  Concrete state models are obtained by having the value set be the set of \gil values, $\vals$ and the set of interpreted symbols be the set of \gil values, $\vals$. Symbolic state models are obtained by having the value set be the set of \gil symbolic expressions, $\sexps$ and the set of interpreted symbols be the set of symbolic variables, $\lxs$. The expression evaluation function is defined in the standard way for both. \end{definition}

Onward, we denote concrete state models by $\gstate(\actions, \cmemory(\actions)) \defeq \langle\sset{\gstate}, \kwT{ea} \rangle$, and symbolic state models by $\hat\gstate(\actions, \smemory(\actions)) \defeq \langle\sset{\hat\gstate}, \kwT{ea} \rangle$. Similarly to memory models, the notation for concrete state models can be simplified by taking advantage of the fact that the only
well-formed context is $\true$. For example, concrete states can be
viewed as triples, $\tup{\cmem, \sto, \arec}$, and the action rules that manipulate contexts simplify to the following:

\medskip
\noindent
\begin{tabular}{l@{~}c@{~}l}
  $\kwT{ea}(\st, \kwT{assume}, \vl)$
    & $\semeq$ & 
    $\lbrace (\st, \vl)^\osucc \mid \vl = \true \rbrace$ \\
 $\kwT{ea}(\st, \act, \vl)$
    & $\semeq$ & 
    $\lbrace (\tup{\cmem', \sto, \arec}, \vl')^r \mid
                 \uaction{\cmem}{\act}{\vl}{(\cmem', \vl')^r} 
\rbrace$
\end{tabular}

\medskip
Note, however, that concrete state action execution is not deterministic as a whole, as allocation of interpreted symbols is effectively creating a random \gil value.

%% file: sections/gil.tex
\subsection{\gil Execution Semantics}
\label{sub:execsem}
We define the \gil  single-trace and collecting execution semantics
parametrically using an arbitrary execution state model,
$\gstate(\gval, \mathsf{I}, \actions, \cmemory(\gval, \actions), \allocator(\gval)) \defeq \langle\sset{\gstate}, \sinv, \actions \uplus \sactions, \kwT{ea} \rangle$.
The semantics is formulated using \emph{call stacks}, \emph{outcomes}, and \emph{configurations}. 

\medskip
\begin{display}{\gil Semantic Domains for $\gstate(\gval, \mathsf{I},
    \actions,  \cmemory(\gval, \actions), \allocator(\gval)) \defeq
    \langle\sset{\gstate}, \ceval{~}{},  \kwT{ea} \rangle$}
\begin{tabular}{lr@{\ \ }c@{\ \ }l@{\quad}l}
   Call stacks: & $\cs \in \css{} $ & $\defeq$ & $ \tup{\f} \mid \tup{\f, \x, \sto, i} \cons \cs$ \\
   Outcomes: & $\outcome \in \outcomes$  & $\defeq$ &  $\cont \mid \onormal{\gv} \mid \ofail{\gv} \mid \omissx{\gv}$ \\ 
   Configurations: & $\cf \in \cfs{}$    & $\defeq$ & $\tup{\st, \cs, i}^o$ \\
 \end{tabular}
\end{display}

Call stacks capture the control flow of \gil programs and consist of a non-empty list of stack frames, where the final frame contains only the identifier of the top-level procedure, and the other frames contain the identifier $f$ of the executing procedure, the variable $x$ to which the return value of $f$ will be assigned, the calling store $\sto$, and the index from which the execution of the calling procedure should continue.
An execution of a \gil command produces an outcome, which can either be: $\cont$, meaning that the execution can continue; $\onormal{\gv}$, meaning that the program terminated successfully with return value $\gv$; or $\ofail{\gv}$, meaning that the  program terminated with a non-correctible error described by~$\gv$; or $\omissx{\gv}$, meaning that the program terminated with a missing information error described by~$\gv$.\footnote{The missing information outcome is not relevant for whole-program  execution, but is for compositional execution (cf.~\S\ref{sec:gilspec}). We include it in this presentation for completeness.}
Finally, \gil commands manipulate configurations, which comprise a state, a call stack, the index of the next command to be executed, and the outcome of the last executed~command.

\begin{definition}[Final Configurations]
A configuration $\cf$  is \emph{final} iff its outcome does not equal the
continuation outcome: that is, 
$
\isfinal(\cf) \Longleftrightarrow (\cf = \tup{\st,
    \cs, i}^o \wedge o \neq \cont)
$. A set of configurations, $\sscf$,  is final iff all of its elements are final:
$
\isfinal(\sscf) \Longleftrightarrow \forall \cf \in \sscf.~\isfinal(\cf)
$.
\end{definition}

To simplify the proofs, we introduce two forms of state action composition.

\begin{definition}[State Action Composition]
\label{def:act:comp} 
State actions have two composition operators: 
$$
\small
\begin{array}{r@{~}c@{~}l}
\action{\st}{(\acomp{\act_1}{\act_2})}{\gv}{(\st', \gv')^\result} & \iff &
 (\action{\st}{\act_1}{\gv}{(\st', \gv')^\result} \land r \neq \osucc)~\lor \\ &&
 (\exists \st'', \gv'' . \, \action{\st}{\act_1}{\gv}{(\st'', \gv'')^\osucc} 
   \ \wedge \ 
   \action{\st''}{\act_2}{\gv''}{(\st', \gv')^\result}) 
 \\[2pt]
\action{\st}{(\lstacomp{\act_1}{\act_2})}{[\gv_1, \gv_2]}{(\st', [\gv_1', \gv_2'])^\result} & \iff & 
  (\action{\st}{\act_1}{\gv_1}{(\st', \gv_1')^\result} \land r \neq \osucc \land \gv_2' = \bot)~\lor \\ &&
 (\exists \st'' . \, \action{\st}{\act_1}{\gv_1}{(\st'', \gv_1')^\osucc} 
   \ \wedge \ 
   \action{\st''}{\act_2}{\gv_2}{(\st', \gv_2')^\result} )
\end{array}
$$
where $\bot$ is a dedicated uninterpreted symbol denoting
undefined.
\end{definition}

\begin{definition}[\gil Single-Trace Execution Semantics]
The \gil one-step single-trace semantic relation,
written $\fullsemtrans{\st, \cs, i}{\st', \cs', j}{}{\prog}{\outcome'}{\outcome}$ is defined by the rules given in Figure~\ref{fig:gilrules}, and its multi-step closure and to-termination closure are given below. To minimise clutter, in the rules we elide the continuation outcome, the success result of action execution, as well as the resulting state or value if they are not used further. We also use the notation $\prog \vdash \cf \semarrow_c \cf'$ and $\prog \vdash \cf \semarrow^{n}_C \cf'$ to expose the executed command(s), defined straightforwardly. Note that the $c$ of the one-step transition is uniquely determined, whereas the $C$ of the $n$-step transition need not be.
\begin{mathpar}
\inferrule[Single: Multi-Step Reflexive]
{
}
{
\prog \vdash \cf \semarrow^0 \cf
}
\and
\inferrule[Single: Multi-Step Transitive]
{
\cf \semarrow \cf''
\quad 
\cf'' \semarrow^n \cf'
}
{
\prog \vdash \cf \semarrow^{n+1} \cf'
}
\and
\inferrule[Single: To-Termination]
{
\cf \semarrow^n \cf'
\quad \isfinal(\cf')
}
{
\prog \vdash \cf \ssemarrowfin \cf'
}
\end{mathpar}
\end{definition}

\begin{figure}[!t]
\begin{mathpar}
\small
\hspace*{-0.3cm}
\inferrule[\textsc{Assignment}]
  {
    \cmd(\prog, \cs, i) = \x := \e
    \\\\
   \action{\st}{(\acomp{\kwT{eval}_{\e}}{\kwT{setVar}_{\x}})}{-}{\st'} 
  }{\semtrans{\st, \cs, i}{\st', \cs, i{+}1}{}{\prog}{}}  	 
 \and
\inferrule[\textsc{IfGoto - True}]
  {
    \cmd(\prog, \cs, i) = \ifgoto{\e}{j}
    \\\\
     \action{\st}{(\kwT{assume} \circ \kwT{eval}_{\e})}{-}{\st'}
  }{\semtrans{\st, \cs, i}{\st', \cs, j}{}{\prog}{}}	 
\and
\inferrule[\textsc{IfGoto - False}]
  {
    \cmd(\prog, \cs, i) = \ifgoto{\e}{j}
    \\\\
       \action{\st}{(\kwT{assume} \circ \kwT{eval}_{\neg \e})}{-}{\st'}
  }{\semtrans{\st, \cs, i}{\st', \cs, i+1}{}{\prog}{}}  	 
\\
\inferrule[\textsc{Memory Action - Success}]
  {
    \cmd(\prog, \cs, i) = \x := \act(\e) \\\\
    \action{\st}{((\kwT{setVar}_\x \circ \act) \circ \kwT{eval}_{\e})}{-}{\st'}
  }{\semtrans{\st, \cs, i}{\st', \cs, i{+}1}{}{\prog}{}}  
\quad
\inferrule[\textsc{Memory Action - Error}]
  {
    \cmd(\prog, \cs, i) = \x := \act(\e) \\\\
    \action{\st}{(\act \circ \kwT{eval}_{\e})}{-}{(\st', \gv')^\result}
   \quad \result \neq \osucc
  }{\semtrans{\st, \cs, i}{\st', \cs, i}{}{\prog}{\ofail{\gv'}}} 
\quad
\inferrule[\textsc{uSym}/\textsc{iSym}]
  {
  	  \kwT{sym} \in \{ {\kwT{uSym}, \kwT{iSym}} \}
	  \\\\
      \cmd(\prog, \cs, i) = \x := \kwT{sym}_{}(n)
      \\\\ 
      \action{\st}{(\kwT{setVar}_{\x} \circ \kwT{sym})}{n}{\st'}
  }{
    \semtrans{\st, \cs, i}{\st', \cs, i{+}1}{}{\prog}{}
  }
 \and 
\inferrule[\textsc{Call}]
  {
    \cmd(\prog, \cs, i) = \x := \e(\e')
    \\\\
     \action{\st}{(\lstacomp{\kwT{eval}_{\e}}{(\lstacomp{\kwT{eval}_{\e'}}{\kwT{getStore}})})}{-}{(\st', [\f, [\gv, \sto']])}
    \\\\
    \cs' = \tup{\f, \x, \sto', i{+}1} \cons \cs
   \quad
   \action{\st'}{\kwT{setStore}}{\btup{[\f.\kwT{arg}, \gv]}}{\st''}
  }{\semtrans{\st, \cs, i}{\st'', \cs', 0}{}{\prog}{}}  	
\and
\inferrule[\textsc{Return}]
  {
     \cmd(\prog, \cs, i) = \return{\e}
      \\\\
      \cs = \tup{-, \x, \sto, j} \cons \cs'
     \quad
       \action{\st}{\kwT{eval}_{\e}}{-}{(\st', \gv)}
    \\\\ 
      \action{\st'}{(\lstacomp{\kwT{setStore}}{\kwT{setVar}_{\x}})}{[\sto, \gv]}{\st''}
  }{
    \semtrans{\st, \cs, i}{\st'', \cs', j}{}{\prog}{}
  }
  \and
 \inferrule[\textsc{Top Return}]
  {
     \cmd(\prog, \cs, i) = \return{\e}
      \\\\
    \cs = \tup{\fid} \quad \action{\st}{\kwT{eval}_{\e}}{-}{(\st', \gv)}
  }{
    \semtrans{\st, \cs, i}{\st', \cs, i}{}{\prog}{\onormal{\gv}}
  } 
  \and
 \inferrule[\textsc{Fail}]
  {
     \cmd(\prog, \cs, i) = \fail{e}
      \\\\
   \action{\st}{\kwT{eval}_{\e}}{-}{(\st', \gv)}
  }{
    \semtrans{\st, \cs, i}{\st', \cs, i}{}{\prog}{\ofail{\gv}}
  } 
\end{mathpar}
\vspace*{-0.2cm}
\caption{GIL Execution Semantics}
\vspace*{-0.2cm}
\label{fig:gilrules}
\end{figure}

\begin{definition}[GIL Collecting Semantics]
The one-step, multi-step, and to-termination collecting transition of the GIL execution semantics, which explicitly accounts for all possible branchings, are defined in Figure~\ref{fig:colsem}. Observe that if $\prog \vdash \{ \cf \} \ssemarrowfin \sscf$ holds, this effectively means that the code executed starting from $\cf$ has no unbounded loops or recursive function calls.
\end{definition}

We conclude by showing that the \gil semantics does not lose information.
We use the following notation: $\getpc{\cf}$ to denote $\getpc{\st}$, where $\cf = \tup{\st, -, -}$; and $\getpc{\sscf}$ to denote $\bigvee_{\cf \in \sscf}(\getpc{\cf})$.

\begin{theorem}
The \gil execution semantics respects the coverage property:
$$
\begin{array}{c}
\prog \vdash \cf \rightarrow \sscf'
 \implies \left(\getpc{\cf} \Rightarrow \getpc{\sscf'} \right)
\end{array}
$$
\end{theorem}

\begin{proof}
The only relevant cases are the goto and the memory action, as they may branch. The goto is fully covered because of the two rules, each of which covers one of the two possible branches (`fixing' the (non-)issue that the $\kwT{assume}$ state action does not have full coverage), and the memory action has full coverage because the corresponding state action has full coverage.
\end{proof}

As a corollary, we are able to obtain coverage  for multi-step and to-termination \gil execution.

\begin{figure}[!t]
\begin{mathpar}%
\inferrule[Config: OStep]
{
\sscf' = \{ \cf' \mid \prog \vdash  \cf  \semarrow \cf' \}
}{
\prog \vdash \cf \rightarrow \sscf'
}
\and
\inferrule[Collection: OStep]
{
\sscf_f = \{\cf \in \sscf \mid \isfinal(\cf)\} \quad
\sscf_{\mathit{nf}} = \{\cf \in \sscf \mid \lnot\isfinal(\cf)\} \neq \emptyset
\quad
\\
\sscf'' = \bigcup_{\cf \in \sscf_{\mathit{nf}}} \{ \sscf'' \mid \prog \vdash \cf \rightarrow \sscf''\} \quad \sscf' = \sscf_f \cup \sscf''
}
{
\prog \vdash \sscf \rightarrow \sscf'
}
\\
\inferrule[Collection: MStep Reflexive]
{
}
{
\prog \vdash \sscf \rightarrow^0 \sscf
}
\and
\inferrule[Collection: MStep Transitive]
{
\sscf \rightarrow \sscf''
\quad 
\sscf'' \rightarrow^n \sscf'
}
{
\prog \vdash \sscf \rightarrow^{n+1} \sscf'
}
\and
\inferrule[Collection: To-Termination]
{
\sscf \rightarrow^n \sscf'
\quad \isfinal(\sscf')
}
{
\prog \vdash \sscf \csemarrowfin \sscf'
}
\end{mathpar}
\vspace*{-0.2cm}
\caption{The \gil Collecting Semantics}
\vspace*{-0.2cm}
\label{fig:colsem}
\end{figure}

\subsection{Correctness: Under-Approximation: True Bug-Finding}
\label{sec:cor:tbf}

We connect the \gil symbolic execution with the \gil concrete execution via \emph{interpretations}, and
formulate a \emph{completeness} property that the symbolic memory actions need to satisfy with respect to concrete memory actions so that the \gil symbolic execution yields no false positive bug reports.

\subsubsection{Under-approximation}
The \gil execution is inherently suitable for under-approximation, \linebreak given the way in which \gil states are constructed:
\begin{propos}[Monotonicity of Contexts] Given a state model and the associated \gil semantics, the following monotonicity property holds for state actions and configurations:
$$
\begin{array}{r@{\implies}l}
\st.\act(\gv) \rightsquigarrow (\st', \gv')^\result & (\access{\st'}{ctx} \Rightarrow \access{\st}{ctx}) \\
\prog \vdash \cf \rightsquigarrow^n \cf' & \access{(\access{\cf'}{st})}{ctx} \Rightarrow \access{(\access{\cf}{st)}}{ctx}
\end{array}
$$
\end{propos}

This corresponds to the well-known property of traditional symbolic execution, in which path conditions can only get strengthened. 


\subsubsection{Treatment of Symbols} Before stating the correctness results for \gillian bug-finding, we address the issue of freshly allocated locations and symbolic variables; this issue arises because these are meant to be created exclusively using allocators (cf.~\S\ref{sub:alloc}). It requires us to restrict the possible memory and state action executions, as follows.

\begin{desiderata}[Symbolic Memory Actions and Symbols] 
\label{req:memsymvar}
If $\faction \smem \act \sexp {} {\smem'} {\sexp'} {\pc'} \result$, then $\smem'$, $\sexp'$ and $\pc'$ must not contain locations/symbolic variables outside of those present in $\smem$ and $\sexp$.
\end{desiderata}

Given this, we can prove the following lemma by straightforward induction and case analysis:

\begin{lemma}[\gil Symbolic Execution and Symbols]
Given a state model whose memory actions satisfy Requirements~\ref{req:memsymvar}, and a symbolic configuration $\scf = \tup{\sst, \hat\cs, -}$ such that all locations/symbolic variables in  $\sst$ and in $\hat\cs$ have been allocated by $\sst$, and if $\prog \vdash \scf \rightsquigarrow^n \scf'$, then:
\begin{itemize}
\item all locations/symbolic variables in  $\sst'$ and in $\hat{\cs'}$ have been allocated by $\sst'$, for $\scf = \tup{\sst', \hat{\cs'}, -}$;
\item all executed state actions, $\sst.\act(\sexp) \rightsquigarrow (\sst', \sexp')^\result$, respect the property that $\sexp$ does not contain locations/symbolic variables outside of those allocated by $\sst$, and $\sexp'$ does not contain loca-tions/symbolic variables outside of those allocated by $\sst'$.
\end{itemize}
\end{lemma}

Therefore, we require the following when considering state actions in isolation:

\begin{desiderata}[Symbolic State Actions and Symbols]
\label{req:stasymvar}
If
$\sst.\act(\sexp) \rightsquigarrow (\sst', \sexp')^\result$, then $\sexp$ 
must not contain locations/symbolic variables outside of those present in $\sst$. 
\end{desiderata}

\subsubsection{Interpretations}
We relate symbolic values and states to their concrete counterparts via \linebreak \emph{interpretations}. At the core of interpretations are symbolic environments, $\varepsilon : \lxs \pmap \sexps$, which map symbolic variables to concrete values.
Given a symbolic environment $\varepsilon$, a \emph{value interpretation}, $\sintp : \sexps \pmap \vals$, is then obtained inductively:
$$
\begin{array}{c}
\sintp{(\vl)}{}  \defeq  \vl \qquad
\sintp{(\lx)}{}  \defeq  \varepsilon(\lx) \qquad
\sintp{(\unop{\sexp})}{}  \defeq  \ceval{\unop{\sintp{(\sexp)}{}}}{} \qquad
\sintp{(\binop{\sexp_1}{\sexp_2}}){}  \defeq   \ceval{\binop{\sintp{(\sexp_1)}{}}{\sintp{(\sexp_2)}{}}}{\emptyset}
\end{array}
$$
where the underline notation for unary and binary operators denotes concrete evaluation. Memory interpretation, $\sintp : \sset{\smemory} \pmap \sset{\cmemory}$, is parametric on the value interpretation~$\varepsilon$ and needs to be defined per instantiation. Value interpretations are then extended to stores, allocators, call stacks, and outcomes in the standard way, overloading notation.

A \emph{state interpretation}, $\sintp : \sset{\hat\gstate} \pmap
\sset{\gstate}$, maps symbolic to concrete states, as follows:
$$
     \sintp(\tup{\smem, \ssto, \sarec, \pc}) \defeq
     {\begin{cases}
     	\tup{\sintp(\smem), \sintp(\ssto), \sintp(\sarec)}, & \text{$\sintp(\smem)$, $\sintp(\ssto)$, $\sintp(\sarec)$ defined, $\sintp(\pc) = \true$} \\
	   \text{not defined}, & \text{otherwise}
	 \end{cases}}
$$
and is extended to a \emph{configuration interpretation} as follows:
$$\sintp(\tup{\sst, \hat\cs, i}) = \tup{\sintp(\sst), \sintp(\hat\cs), i},$$
being undefined if ${\sintp(\sst)}$ is undefined. With these definitions in place, the \emph{models} of a symbolic state $\sst$ are defined in the standard way:
$$
\modls(\sst) = \{ \sintp(\sst) \mid \Vint\sintp\sst \},
$$
and are extended to configurations as follows:
$$
\modls(\tup{\sst, \hat\cs, i}^o) = \{ \tup{\sintp(\sst), \sintp(\hat\cs), i}^{\sintp(o)} \mid \Vint\sintp\sst \}.
$$

\begin{lemma}[Monotonicity of Interpretations]
Given a state model and the associated \gil semantics, the following monotonicity property holds for interpretations:
$$
\begin{array}{r@{\implies}l}
\st.\act(\gv) \rightsquigarrow (\st', \gv')^\result \land \st' \in \domain(\sintp) & \st, \gv,  \gv' \in \domain(\sintp) \\
\prog \vdash \cf \rightsquigarrow^n \cf' \land \cf' \in \domain(\sintp) & \cf \in \domain(\sintp)
\end{array}
$$
\end{lemma}
\begin{proof}
Trivial, given symbolic variable requirements on memory and state actions.
\end{proof}

For all correctness results, the following requirements needs to hold:
\begin{desiderata}[Well-Formedness and Interpretations]\label{des:memwf}
Given a concrete and a symbolic memory model, $\cmemory(A) = \tup{\sset{\cmemory}, \sinv, \cea}$ and $\smemory(A) = \tup{\sset{\smemory}, \hat\sinv, \sea}$, interpretations must preserve well-formedness:
$$
\sinv_\pc(\smem) \land \smem \in \domain(\sintp) \land \sintp(\pc) = \true \implies \sinv(\sintp(\smem))
$$
\end{desiderata}

It is straightforward to prove, given these requirements, that state and configuration interpretation also preserve well-formedness.

\subsubsection{Backward Completeness} Our goal for the core symbolic execution of \gillian is to prove that it generates no false positive bug reports. We capture this with the following property:

\begin{theorem}[Backward Completeness: GIL]\label{des:bc:gil}
Given a concrete and a symbolic state model, $\gstate(\actions, \cmemory(\actions)) \defeq \langle\sset{\gstate}, \kwT{ea} \rangle$, and  $\hat\gstate(\actions, \smemory(\actions)) \defeq \langle\sset{\hat\gstate}, \kwT{ea} \rangle$, respectively, the coresponding \gil symbolic execution obeys the following property:
\begin{mathpar}
\small
          \inferrule[GIL - Backward Completeness (GIL-BC)]{}{
        {\begin{array}{l}
\prog \vdash  \scf \semarrow^n_C \scf' \land
      \cf' \in \modls(\hat\cf')
        \implies \exists \cf.~\cf \in \modls(\hat\cf) \land \prog \vdash  \cf \semarrow^n_C {\cf'}
                \end{array}}}
\end{mathpar}
\end{theorem}

The backward completeness property states that the set of concrete
 states reachable from $\modls({\hat\cf})$ by
executing ${C}$ is \emph{under-approximated} by $\modls(\scf')$,
meaning that any concrete state in $\modls(\scf')$ is reachable
by executing $C$ from some concrete state in 
$\modls({\hat\cf})$. This property resonates strongly with the meaning of incorrectness logic triples~\cite{incohearn,isl}.
 If this property holds, this means that any
symbolic trace that results in an error state has a corresponding
concrete trace that also results in an error state: that is,
that all of the bugs found by \gillian are true.

We proceed by formulating requirements on memory actions, which then lift to state actions, which then, in turn, lift to the \gil symbolic execution, yielding the required property, as per Requirements~\ref{des:bc:gil}.

\begin{desiderata}[Backward Completeness: Memory Actions]\label{des:memact}
Given a concrete and a symbolic memory model, $\cmemory(A) = \tup{\sset{\cmemory}, \sinv, \cea}$ and $\smemory(A) = \tup{\sset{\smemory}, \hat\sinv, \sea}$, the following  property must hold for all
memory actions $\act \in \mactions$:
\begin{mathpar}
\inferrule[Memory Action - Backward Completeness (MA-BC)]
  {}{
  {\hspace*{-0.15cm}\begin{array}{l}
\faction{\smem}{\act}{\sexp}{\pc}{\smem'}{\sexp'}{\pc'}\result
  \land \sint(\pc') = \true \land \smem, \sexp, \pc \in \domain(\sintp)
  \\ \qquad \implies
        \caction{\sintp(\smem)}{\act}{\sintp({\sexp})}{(\sintp(\smem'), \sintp({\sexp'}) )^\result}
        \end{array}}
 }
\end{mathpar}
\end{desiderata}

These requirements have to be proven by the tool developer for each \gillian instantiation. Observe that they are slightly stronger than the analogous ones for the \gil symbolic semantics given in Requirements~\ref{des:bc:gil}, in that the witness, in this case $\sintp(\smem)$, is explicitly provided.
Note also that the interpretation $\sintp$ has to realise $\pc'$, which would be, in effect, the context of the symbolic state after the action execution, as per Definition~\ref{def:esm}. We can now prove the following property for symbolic state actions by case analysis:

\begin{lemma}[Backward Completeness: State Actions]\label{des:stateact}
Given a concrete and a symbolic state model, $\gstate(\actions, \cmemory(\actions)) \defeq \langle\sset{\gstate}, \kwT{ea} \rangle$, and  $\hat\gstate(\actions, \smemory(\actions)) \defeq \langle\sset{\hat\gstate}, \kwT{ea} \rangle$, respectively,
if the memory actions satisfy Requirements~\ref{des:memact}, then the following property holds for all state actions $\act \in \actions \uplus \sactions$:
\begin{mathpar}
%
          \inferrule[State Action - Backward Completeness (SA-BC)]{}{
          {\begin{array}{l}
        \action{\sst}{\act}{\sexp}{(\sst', \sexp')^\result} \land \sst' \in \domain(\sintp) 
        \\
        \qquad  \implies  \action{\sint(\sst)}{\act}{\sintp(\sexp)}{(\sintp(\sst'), \sintp(\sexp'))^\result}
        \end{array}}
        }
  \end{mathpar}
\end{lemma}

\begin{proof}
By case analysis on $\alpha$. The $\kwT{setVar}_{\x}$, $\kwT{setStore}$, $\kwT{getStore}$, and $\kwT{eval}_{\e}$ cases are straightforward, noting that $\ceval{\e}{\sint(\ssto)} = \sint({\ceval{\e}\ssto})$ (proven by induction on $\e$). 
We address the non-trivial cases, which alter the context and the allocators. 
We have the hypotheses: \hyp{1}~$\action{\sst}{\act}{\sexp}{(\sst', \sexp')^\result}$; and \hyp{2}~$\sst' \in \domain(\sintp)$. By Requirements~\ref{req:stasymvar}, we have then that also \hyp{3}~$\sst, \sexp, \sexp' \in \domain(\sintp)$.

\medskip
\case{\prooflab{Assume}}
\begin{proofpf}
  \pflongnumbers
  \step{s2} {[by the hypotheses]
    \begin{pfenum}
  	  \item $\hat\st = \tup{\smem, \ssto, \sarec, \pc}$
	  \item $\sexp = \pc'$
  	  \item $\hat\st' = \tup{\smem, \ssto, \sarec, \pc \land \pc'}$
	  \item $\pc \land \pc'~\sat$
	  \item $\sexp' = \true$
	  \item $\result = \osucc$
    \end{pfenum} }
  \step{s3} {$\sintp(\sst') = \sintp(\sst)$ [by H2, monotonicity of interpretation, 1.1 and 1.3]}
  \step{s4} {$\sint({\pc'}) = \true$ [by H2]}
  \step{sx} {$\action{\sint(\sst)}{\kwT{assume}}{\seval{\pc'}\lenv}{(\sintp(\sst'), \sintp(\sexp'))^\osucc$} (G1) [by definition of $\kwT{assume}$ and previous steps]}
\end{proofpf}

\case{\prooflab{uSym}}
\begin{proofpf}
  \pflongnumbers
  \step{s2} {[by the hypotheses]
    \begin{pfenum}
  	  \item $\hat\st = \tup{\smem, \ssto, \sarec, \pc}$
	  \item $\alloc{\sarec}{n}{\sarec', \lst\loc}{\clocs}$
  	  \item $\hat\st' = \tup{\smem, \ssto, \sarec', \pc}$
	  \item $\sexp' = \lst\loc$ 
	  \item $\result = \osucc$
    \end{pfenum} }
  \step{s4} {$\alloc{\sint(\hat\arec)}{n}{\sint({\hat\arec'}) , \sintp(\lst\loc)}{\clocs}
$ [by 1.2 and the fact that interpretations do not affect locations]}
  \step{sx} {$\action{\sint(\sst)}{\kwT{uSym}}{n}{(\sintp(\sst'), \sint(\sexp'))^\osucc$} (G1) [by definition of $\kwT{uSym}$ and previous steps]}
\end{proofpf}

\case{\prooflab{iSym}}
\begin{proofpf}
  \pflongnumbers
  \step{s2} {[by the hypotheses]
    \begin{pfenum}
  	  \item $\hat\st = \tup{\smem, \ssto, \sarec, \pc}$
	  \item $\alloc{\sarec}{n}{\sarec', \lst\lx}{\lxs}$
  	  \item $\hat\st' = \tup{\smem, \ssto, \sarec', \pc}$
	  \item $\sexp' = \lst\lx$ 
	  \item $\result = \osucc$
    \end{pfenum} }
  \step{s4} {$\alloc{\sint(\hat\arec)}{n}{\sint({\hat\arec'}) , \sintp(\lst\lx)}{\vals}
$ [by 1.2 and the fact that any \gil values can be allocated]}
  \step{sx} {$\action{\sint(\sst)}{\kwT{iSym}}{n}{(\sintp(\sst'), \sint(\sexp')^\osucc$} (G1) [by definition of $\kwT{uSym}$ and previous steps]}
\end{proofpf}

\case{\prooflab{Memory Action}}
\begin{proofpf}
  \pflongnumbers
  \step{s2} {[by Hs]
    \begin{pfenum}
  	  \item $\hat\st = \tup{\smem, \ssto, \sarec, \pc}$
	  \item $\faction{\smem}{\act}{\sexp}{\pc}{\cmem'}{\sexp'}{\pc'} \result$
  	  \item $\hat\st' = \tup{\smem', \ssto, \sarec, \pc'}$
    \end{pfenum} }
  \step{sx1} {$\hat\sinv_\pc(\smem)$ [by H2]}
  \step{sx2} {$\sintp(\pc') = \true$ [by H2]}
  \step{sx3} {$\smem, \sexp \in \domain(\sintp)$ [by H2 and Requirements~\ref{req:stasymvar}]}
      \step{sx4} {$\caction{\sint(\smem)}{\act}{\sint({\sexp})}{(\sint(\smem'), \sint({\sexp'}) )}^\result
$ [by \prooflab{MA-C}, 2, 1.2, 4 and 5]}
  \step{sx} {$\action{\sint(\sst)}{\act}{\sint(\sexp)}{(\sint(\sst'), \sint(\sexp'))^\result$} (G1) [by definition of memory action and previous steps]}
\end{proofpf}
\end{proof}

The next part to prove is that backward completeness is preserved by state action composition:

\begin{lemma}[Action Composition]
\label{lem:sactsc}
If two state actions $\act_1, \act_2 \in \actions$ 
satisfy \prooflab{SA-BC}, then their compositions, $\acomp{\act_1}{\act_2}$ and $\lstacomp{\act_1}{\act_2}$,
 also satisfy \prooflab{SA-BC}.  
\end{lemma}

\begin{proof}
For both compositions, the cases in which $\act_1$ results in an error trivially hold, as they amount to the appropriate property of $\act_1$; we prove the cases in which the first action succeeds.
We first prove the property for $\acomp{\act_1}{\act_2}$. 

\begin{proofpf}
     \pflongnumbers
 \step{1}{
      \assume{
        \begin{pfenum}
          \item (H1) $\action{\sst}{(\acomp{\act_1}{\act_2})}{\sexp}{(\sst', \sexp')^\result}$
          \item (H2) $\Vint{\sint}{\sst''}$
        \end{pfenum}
      }
    }
      \step{2}{$ \exists \sst_1, \sexp_1$ [From Action Composition on H1]}
        \begin{pfenum*}
           \item $\action{\sst}{\act_1}{\sexp}{(\sst_1, \sexp_1)^\osucc}$  \label{1m1}
           \item $ \action{\sst_1}{\act_2}{\sexp_1}{(\sst', \sexp')^\result}$ \label{1m2}
      \end{pfenum*}
      \step{3}{$\Vint\sint{\sst_1}$ [by 2.2 and monotonicity of interpretations]}
      \step{4}{$\action{\sint(\sst)}{\act_1}{\sintp(\sexp)}{(\sintp(\sst_1), \sintp(\sexp_1))^\osucc$} [by [SA-BC] applied to 2.1, and 3]}
      \step{6}{$\action{\sint(\sst_1)}{\act_2}{\sintp(\sexp_1)}{(\sintp(\sst'), \sintp(\sexp'))^\result$} [by [SA-BC] applied to 4 and H2]}
      \step{7}{$\action{\sintp(\sst)}{(\acomp{\act_1}{\act_2})}{\sintp(\sexp)}{(\sintp(\sst'), \sintp(\sexp'))^\result}$ [by 4 and ]}
          \end{proofpf} 
  
  \noindent We conclude by proving the property for $\lstacomp{\act_1}{\act_2}$. 
\begin{proofpf}
     \pflongnumbers
 \step{1}{
      \assume{
        \begin{pfenum}
          \item (H1) $\action{\sst}{(\lstacomp{\act_1}{\act_2})}{[\sexp_1, \sexp_2]}{(\sst', [\sexp_1', \sexp_2'])^\result}$
          \item (H2) $\Vint{\sint}{\sst''}$
        \end{pfenum}
      }
    }
      \step{2}{$ \exists \sst_1, \sexp_1$ [From Action Composition on H1]}
        \begin{pfenum*}
           \item $\action{\sst}{\act_1}{\sexp_1}{(\sst_1, \sexp_1')^\osucc}$  \label{1m1}
           \item $ \action{\sst_1}{\act_2}{\sexp_2}{(\sst', \sexp_2')^\result}$ \label{1m2}
      \end{pfenum*}
      \step{3}{$\Vint\sint{\sst_1}$ [by 2.2 and monotonicity of interpretations]}
      \step{4}{$\action{\sint(\sst)}{\act_1}{\sintp(\sexp_1)}{(\sintp(\sst_1), \sintp(\sexp_1'))^\osucc$} [by [SA-BC] applied to 2.1 and 3]}
      \step{6}{$\action{\sint(\sst_1)}{\act_2}{\sintp(\sexp_2)}{(\sintp(\sst'), \sintp(\sexp_2'))^\result$} [by [SA-BC] applied to 4 and H2]}
      \step{7}{$\action{\sintp(\sst)}{(\lstacomp{\act_1}{\act_2})}{[\sintp(\sexp_1), \sintp(\sexp_2)]}{(\sintp(\sst'), [\sintp(\sexp_1'), \sintp(\sexp_2')])^\result}$ [by 4 and 5]}
          \end{proofpf} 
\end{proof}

We can now move on to \gil execution:

\begin{lemma}[Backward Completeness: \gil One-Step]\label{lem:gilbcone}
Given a concrete and a symbolic state model, $\gstate(\actions, \cmemory(\actions)) \defeq \langle\sset{\gstate}, \kwT{ea} \rangle$, and  $\hat\gstate(\actions, \smemory(\actions)) \defeq \langle\sset{\hat\gstate}, \kwT{ea} \rangle$, respectively,
if the memory actions satisfy Requirements~\ref{des:memact}, then the corresponding \gil one-step symbolic semantics satisfies the following:
\begin{mathpar}
\small
          \inferrule[\gil One-Step - Backward Completeness (\gil-OS-BC)]{}{
          {\begin{array}{l}
\prog \vdash \hat\cf \semarrow_c \hat\cf' \gand \hat\cf' \in \domain(\sintp) 
        \implies \prog \vdash \sint({\hat\cf}) \semarrow_c  \sint({\hat\cf'}) 
        \end{array}}
        }
  \end{mathpar}
\end{lemma}

\begin{proof}
The cases which involve only one state action execution or only one action composition follow trivially from Lemma~\ref{des:stateact} and Lemma~\ref{lem:sactsc}. The only two remaining cases are the function call rule and the return rule. We prove the function call rule; the return rule is proven analogously:
\case{\prooflab{Call}}
\begin{proofpf}
  \pflongnumbers
  \step{s1} {Hypotheses:
    \begin{pfenum}
      \item[(H1)] $\prog \vdash \hat\cf \semarrow_c \hat{\cf'}$
      \item[(H2)] $\Vint\sint{\hat\cf'}$
    \end{pfenum}}
  \step{s2} {[by H1, \prooflab{Call} case]
    \begin{pfenum}
  	  \item $\hat\cf = \tup{\sst, \hat\cs, i}$
  	  \item $\cmd(\prog, \hat\cs, i) = \x := \e(\e')$
	  \item $\action{\sst}{(\lstacomp{\kwT{eval}_{\e}}{(\lstacomp{\kwT{eval}_{\e'}}{\kwT{getStore}})})}{-}{(\sst_1, [\f, [\sexp, \ssto']])}$
	  \item $\action{\sst_1}{\kwT{setStore}}{\btup{[\f.\kwT{arg}, \sexp]}}{(\sst', -)}$
	  \item $\hat\cf' = \tup{\sst', \tup{\f, \x, \ssto', i{+}1} \cons \hat\cs, 0}^\cont$
    \end{pfenum} }
%
%
  \step{s4b} {$\Vint\sint{\sst'}$ [by H2 and definition of interpretation for configurations]}
   \step{4b}{$\Vint\sint{{\sst_1}}$ [Applying monotonicity of interpretations to 2.4 + 3]}
  \step{s5} {$\action{\seval{\sst}{\lenv}}{(\lstacomp{\kwT{eval}_{\e}}{(\lstacomp{\kwT{eval}_{\e'}}{\kwT{getStore}})})}{-}{(\sint({\sst_1}), [\f, [\sint(\sexp), \sint(\ssto')]])}$ [by \prooflab{SA-BC} +  Lemma~\ref{lem:sactsc} + \stepref{s2}.3]}
  \step{xx} {$\action{\sint({\sst_1})}{\kwT{setStore}}{\btup{[\f.\kwT{arg}, \sint(\sexp)]}}{(\sint({{\sst'}{}}), -)}$
   [by \prooflab{SA-BC}, given \stepref{s2}.4 + 4]\footnote{We are overloading the notation for an actual store and its serialised representation.}}
  \step{s6} {$\cmd(\prog, \sint(\hat\cs), i) = \x := \e(\e')$ [by \stepref{s2}.2 + definition of $\sint(\hat\cs)$]}
  \step{s7} {$\semtrans{\sint(\sst), \sint(\hat\cs), i}{\sint({\sst'}), \tup{\f, \x, \sint(\ssto'), i{+}1} \cons \sint(\hat\cs), 0}{}{\prog}{\cont}$ [by \stepref{s6} + \stepref{s5} + 6]}  
  \step{s8} {$\prog \vdash \sint({\hat\cf}) \semarrow  \sint({\hat\cf'})$ (G1) [by \stepref{s7} + \stepref{s2}.1 + \stepref{s2}.5 + definition of interpretation for configurations]}
\end{proofpf}

\end{proof}

This property can be easily generalised to an $n$-step execution:

\begin{lemma}[Backward Completeness: \gil Multi-Step]\label{lem:gilbcmult}
Given a concrete and a symbolic state model, $\gstate(\actions, \cmemory(\actions)) \defeq \langle\sset{\gstate}, \kwT{ea} \rangle$, and  $\hat\gstate(\actions, \smemory(\actions)) \defeq \langle\sset{\hat\gstate}, \kwT{ea} \rangle$, respectively,
if the memory actions satisfy Requirements~\ref{des:memact}, then the corresponding \gil multi-step symbolic semantics satisfies the following:
\begin{mathpar}
\small
          \inferrule[\gil One-Step - Backward Completeness (\gil-OS-BC)]{}{
          {\begin{array}{l}
\prog \vdash \hat\cf \semarrow^n_C \hat\cf' \gand \hat\cf' \in \domain(\sintp) 
        \implies \prog \vdash \sint({\hat\cf}) \semarrow^n_C  \sint({\hat\cf'}) 
        \end{array}}
        }
  \end{mathpar}
\end{lemma}

\begin{proof} By induction on $n$. The case when $n = 0$ is trivial; assume $n = k+1$.
\begin{proofpf}
  \pflongnumbers
    \step{2ass}{
      \begin{pfenum}
        \item[(H1)] $\prog \vdash \hat\cf \semarrow^{k+1} \hat\cf'$
        \item[(H2)] $\Vint\sint{\hat\cf'}$
      \end{pfenum}
    }
		\step{2a}{$\exists \hat\cf'', c, C'.$}
		\begin{pfenum*}
			\item $\prog \vdash \hat\cf \semarrow_c \hat\cf''$ [From H2]  \label{2a1}
			\item $\prog \vdash \hat\cf'' \semarrow^k_{C'} \hat\cf'$ \label{2a2} [From H2]
			\item $C = c :: C'$ \label{2a3}
			\item $\Vint\sint{\hat\cf''}$ [From monotonicity of interpretations] \label{2a5} 
		\end{pfenum*}
		\step{2b}{$\prog \vdash \sint({\hat\cf}) \semarrow_c  \sint(\hat\cf'')$ [by Lemma~\ref{lem:gilbcone}, \stepref{2a}.\ref{2a1}, and \stepref{2a}.\ref{2a5}]}
		\step{2c}{$\prog \vdash \sint({\hat\cf''}) \semarrow^k_{C'}  \sint(\hat\cf')$ [by IH,  \stepref{2a}.\ref{2a2}, and H2]}
		\step{2d}{$\prog \vdash \sint({\hat\cf}) \semarrow^{k+1}_{C} \sint(\hat\cf')$ (G1) [From \stepref{2b}, \stepref{2c}, and \stepref{2a}.\ref{2a3}]} 
	\end{proofpf}
\end{proof}

We can now prove that Theorem~\ref{des:bc:gil} holds. Assume \hyp{1}~$\prog \vdash  \scf \semarrow^n_C \scf'$, and \hyp{2}~$\cf' \in \modls(\hat\cf')$. By definition of models, H2 means that $\cf' =\sintp(\hat\cf')$, for some \hyp{3}~$\hat\cf' \in \domain(\sintp)$. Then, applying Lemma~\ref{lem:gilbcmult} to \hyp{1}, and \hyp{3}, we obtain that $\prog \vdash \sint({\hat\cf}) \semarrow^n_C  \sint({\hat\cf'})$. Given the definition of models, we have that $\sint({\hat\cf}) \in \modls(\hat\cf)$, that is, that $\sint({\hat\cf})$ is the witness in the conclusion of Theorem~\ref{des:bc:gil}. Finally, from Theorem~\ref{des:bc:gil}, we can derive the following, more general, property, using the definition of the \gil collecting semantics:
\begin{mathpar}
           \inferrule[GIL - Collecting Backward Completeness (GIL-CBC)]{}{
        {\begin{array}{l}
\prog \vdash \ssscf \rightarrow^n \ssscf' \gand 
       \cf' \in \modls(\ssscf') 
\implies \exists \cf, m.~\cf \in \modls({\ssscf}) \gand m \leq n \gand \prog \vdash  \cf \semarrow^m {\cf'}
                \end{array}}}
\end{mathpar}


\subsection{Correctness: Forward Soundness}
\label{sec:cor:bc}

In contrast to backward completeness, which requires an existence of a single concrete trace corresponding to a given single symbolic trace, forward soundness states that a given set of symbolic traces covers the set of all corresponding concrete traces.
We first formulate this property for memory actions and then prove that it holds for state actions:

\begin{desiderata}[Forward Soundness: Memory Actions] 
\label{req:mem:fs}
Given a concrete and a symbolic memory model, $\cmemory(A) \defeq \tup{\sset{\cmemory}, \sinv, \cea}$ and $\smemory(A) \defeq \tup{\sset{\smemory}, \sinv, \sea}$, the following soundness property must hold for all
memory actions $\act \in \mactions$:
\begin{mathpar}
\inferrule[Memory Actions: Forward Soundness (MA-FS)]{}{
          {\begin{array}{l}
\fcolactionx \smem \act \sexp {\pc} {\{ (\smem_i, \sexp_i)^{\result_i}_{\pc_i} |_{i \in I}\}}  \gand  \sintp(\pc) = \true \gand 
\action{\sint(\smem)}{\act}{\sintp(\sexp)}{(\cmem', \vl')^\result} \\ \qquad \implies \exists i.~(\cmem', \vl')^\result_\true = \sintp((\smem_i, \sexp_i)^{\result_i}_{\pc_i})
 \end{array}}}
\end{mathpar}
\end{desiderata}

As an interlude, we give the following proposition, which is straightforwardly proven given the coverage property of memory actions:
\begin{propos}[Interpretations and Coverage] 
\label{prop:intermezzo}
$$
\fcolactionx \smem \act \sexp {\pc} {\{ (\smem_i, \sexp_i)^{\result_i}_{\pc_i} |_{i \in I}\}}  \gand  \sintp(\pc) = \true \implies \exists i.~\sintp(\pc_i) = \true
$$ 
\end{propos}

We now proceed to prove forward soundness of state actions.

\begin{lemma}[Forward Soundness: State Actions]
\label{req:st:fs}
Given a concrete and a symbolic state model, $\gstate(\actions, \cmemory(\actions)) \defeq \langle\sset{\gstate}, \kwT{ea} \rangle$, and  $\hat\gstate(\actions, \smemory(\actions)) \defeq \langle\sset{\hat\gstate}, \kwT{ea} \rangle$, respectively,
if the memory actions satisfy Requirements~\ref{req:mem:fs}, then the following property holds for all state actions $\act \in \actions \uplus \sactions$:
\begin{mathpar}
\inferrule[State Actions: Forward Soundness: Interpretations (SA-FS-I)]{}{
  {\begin{array}{l}
\fcolactionx \sst \act \sexp {} \ssst \gand \action{\sint(\sst)}{\act}{\sintp(\sexp)}{(\st', \vl')^\result} \\ \qquad \implies \exists \varepsilon' \geq \varepsilon.~(\st', \vl')^\result \in \sintpe{\varepsilon'}(\ssst)
 \end{array}}}
\end{mathpar}
\end{lemma}

\begin{proof}
The $\kwT{setVar}$, $\kwT{setStore}$, $\kwT{getStore}$, $\kwT{eval}_{\e}$, and $\kwT{assume}$, cases follow trivially from the definitions, as they do not introduce new symbolic variables and are deterministic. We address the non-trivial cases of symbol allocation and memory action. 
We have the hypotheses: 
	\hyp{1}~$\fcolactionx \sst \act \sexp {} \ssst$; and 
	\hyp{2}~$\action{\sint(\sst)}{\act}{\sintp(\sexp)}{(\st', \vl')^\result}$.
	
\medskip
\case{\prooflab{uSym}}
\begin{proofpf}
  \pflongnumbers
    \step{s3} {[by H2 and the definition of concrete $\kwT{uSym}$]
    \begin{pfenum}
  	  \item $\sintp(\sst) = \tup{\sintp(\smem), \sintp(\ssto), \sintp(\sarec), \true}$
	  \item $\sintp(\hat e) = k$
	  \item $\alloc{\sintp(\sarec)}{k}{\arec', \lstls}{\clocs}$ 
  	  \item $\st' = \tup{\sintp(\smem), \sintp(\ssto), \sarec', \true}$
	  \item $\vl' = \lstls$ 
	  \item $\result = \osucc$
    \end{pfenum} }
  \step{s2} {[by H1 and the definition of symbolic $\kwT{uSym}$]
    \begin{pfenum}
  	  \item $\hat\st = \tup{\smem, \ssto, \sarec, \pc}$
	  \item $\hat e = k$, for some $k \in \mathbb{N}$
	  \item $\alloc{\sarec}{k}{\sarec', \lstls}{\clocs}$ 
	  \item $\sst' = \tup{\smem, \ssto, \sarec', \pc} $
	  \item $\{\sst', \lstls\}^\osucc \in \ssst$
    \end{pfenum} }
  \step{s4} {Take $\varepsilon' = \varepsilon$. Then:
      \begin{pfenum}
      \item $\sintpe{\varepsilon'}(\smem) = \sintp(\smem)$, $\sintpe{\varepsilon'}(\ssto) = \sintp(\ssto)$, and $\sintpe{\varepsilon'}(\pc) = \true$ [because $\varepsilon' = \varepsilon$]
      \item $\sintpe{\varepsilon'}(\sarec') = \arec'$ [by 2.3, 1.3, the definition of $\varepsilon'$, and the definition of allocation]
      \item $\st' = \sintpe{\varepsilon'}(\tup{\smem, \ssto, \sarec', \pc})$ [by 1.4, 3.1, and 3.2]
  	  \item $\vl' = \sintpe{\varepsilon'}(\lstls)$ [by 1.5, and definition of $\varepsilon'$]
	  \item (G1) [by 2.4, 3.3, 3.4, and 1.6]
    \end{pfenum}}
\end{proofpf}
\case{\prooflab{iSym}}
\begin{proofpf}
  \pflongnumbers
  \step{s2} {[by H1 and the definition of symbolic $\kwT{iSym}$]
    \begin{pfenum}
  	  \item $\hat\st = \tup{\smem, \ssto, \sarec, \pc}$
	  \item $\hat e = k$, for some $k \in \mathbb{N}$
	  \item $\alloc{\sarec}{k}{\sarec', [\lx_{1} \ldots, \lx_{k}]}{\lxs}$ 
	  \item $\sst' = \tup{\smem, \ssto, \sarec', \pc}$
	  \item $\{\sst', [\lx_{1} \ldots, \lx_{k}]\}^\osucc \in \ssst$
    \end{pfenum} }
  \step{s3} {[by H2 and the definition of concrete $\kwT{iSym}$]
    \begin{pfenum}
  	  \item $\sintp(\sst) = \tup{\sintp(\smem), \sintp(\ssto), \sintp(\sarec), \true}$
	  \item $\sintp(\hat e) = k$
	  \item $\alloc{\sintp(\sarec)}{k}{\arec', [\vl_1 \ldots, \vl_m]}{\vals}$ 
  	  \item $\st' = \tup{\sintp(\smem), \sintp(\ssto), \sarec', \true}$
	  \item $\vl' = [\vl_1 \ldots, \vl_m]$ 
	  \item $\result = \osucc$
    \end{pfenum} }
    
  \step{s4} {Take $\varepsilon' = \varepsilon[\lx_{i} \mapsto \vl_i|_{i=1}^n]$. Then:
      \begin{pfenum}
      \item $\sintpe{\varepsilon'}(\smem) = \sintp(\smem)$, $\sintpe{\varepsilon'}(\ssto) = \sintp(\ssto)$, and $\sintpe{\varepsilon'}(\pc) = \true$ [because $\varepsilon'$ extends $\varepsilon$]
      \item $\sintpe{\varepsilon'}(\sarec') = \arec'$ [by 1.3, 2.3, the definition of $\varepsilon'$, and the definition of allocation]
      \item $\st' = \sintpe{\varepsilon'}(\tup{\smem, \ssto, \sarec', \pc})$ [by 2.4, 3.1, and 3.2]
  	  \item $\vl' = \sintpe{\varepsilon'}([\lx_{1} \ldots, \lx_{k}])$ [by 2.5, and definition of $\varepsilon'$]
	  \item (G1) [by 1.4, 3.3, 3.4, and 2.6]
    \end{pfenum}}
\end{proofpf}

\case{\prooflab{Memory Action}}
\begin{proofpf}
  \pflongnumbers
  \step{s2} {[by H1 and the definition of symbolic memory action]
    \begin{pfenum}
  	  \item $\hat\st = \tup{\smem, \ssto, \sarec, \pc}$
	  \item $\fcolactionx \smem \act \sexp {} {\{ (\smem_i, \sexp_i)^{\result_i}_{\pc_i} |_{i \in I}\}} $
	  \item $\ssst = \{ (\tup{\smem_i, \ssto, \sarec, \pc_i}, \sexp_i)^{\result_i}|_{i \in I} \}  $
    \end{pfenum} }
    \step{s3} {[by H2 and the definition of concrete memory action]
    \begin{pfenum}
  	  \item $\sintp(\sst) = \tup{\sintp(\smem), \sintp(\ssto), \sintp(\sarec), \true}$ 
	  \item $\action{\sint(\smem)}{\act}{\sintp(\sexp)}{(\cmem', \vl')^\result}$
  	  \item $\st' = \tup{\cmem', \sintp(\ssto), \sintp(\sarec'), \true}$
    \end{pfenum} }

  \step{sx2} {$\exists i.~(\cmem', \vl')^\result_\true = \sintp((\smem_i, \sexp_i)^{\result_i}_{\pc_i})$ [by MA-FS applied to  1.2, 2.1, and 2.2]}
  
      \step{sx4} {$\exists \varepsilon' \geq \varepsilon.~(\st', \vl')^\result \in \sintpe{\varepsilon'}(\ssst)$ (G1) [by 2.3, 3, and 1.3, taking $\varepsilon' = \varepsilon$]}
\end{proofpf}
\end{proof}

The next part to prove is that forward soundness is preserved by state action composition:

\begin{lemma}[Action Composition]
\label{lem:sactfs}
If two state actions $\act_1, \act_2 \in \actions$ 
satisfy \prooflab{SA-FS}, then their compositions, $\acomp{\act_1}{\act_2}$ and $\lstacomp{\act_1}{\act_2}$,
 also satisfy \prooflab{SA-FS}.  
\end{lemma}

\begin{proof}
For both compositions, the cases in which $\act_1$ results in an error trivially hold, as they amount to the appropriate property of $\act_1$; we therefore assume, without loss of generality, that $\act_1$ always succeeds.
We first prove the property for $\acomp{\act_1}{\act_2}$. 

\begin{proofpf}
     \pflongnumbers
 \step{1}{
      \assume{
        \begin{pfenum}
          \item (H1) $\fcolactionx \sst {(\acomp{\act_1}{\act_2})} \sexp {} \ssst
          $
          \item (H2) $\action{\sint(\sst)}{(\acomp{\act_1}{\act_2})}{\sintp(\sexp)}{(\st', \vl')^\result}$
        \end{pfenum}
      }
    }
      \step{2}{[From Action Composition on H1 and the assumption that $\alpha_1$ always succeeds]}
        \begin{pfenum*}
           \item $\fcolactionx \sst {\act_1} \sexp {} {\left\{ (\sst_i, \sexp_i)^{\osucc}|_{i \in I}  \right\}} $  \label{1m1}
           \item $(\uaction{\sst_i}{\act_2}{\sexp_i}{\ssst_i})|_{i \in I}$ \label{1m2}
           \item $\ssst = \bigcup_{i \in I} \ssst_i$
      \end{pfenum*}
      \step{3}{[From Action Composition on H2 and the assumption that $\alpha_1$ always succeeds]}
        \begin{pfenum*}
           \item $\action{\sint(\sst)}{\act_1)}{\sintp(\sexp)}{(\st'', \vl'')^\osucc}$
           \item $\action{\st''}{(\act_2)}{\vl''}{(\st', \vl')^\result}$ 
      \end{pfenum*}

      \step{4}{Obtain $\varepsilon'' \geq \varepsilon$, $k \in I$: $(\st'', \vl'') \in \sintpe{\varepsilon''}(\left\{ (\sst_i, \sexp_i)^{\osucc}|_{i \in I}  \right\}) \gand \st'' = \sintpe{\varepsilon''}(\sst_k) \gand \vl'' = \sintpe{\varepsilon''}(\sexp_k)$ [by [SA-FS] applied to 2.1 and 3.1]}
      \step{6b}{Obtain $\varepsilon' \geq \varepsilon''$: $(\st', \vl')^\result \in \sintpe{\varepsilon'}(\ssst_k)$ [by [SA-FS] applied to 4 and 3.2]}
      \step{7}{(G1) [by 5 and 2.3]}
          \end{proofpf} 
  
  \noindent We conclude by proving the property for $\lstacomp{\act_1}{\act_2}$. 
\begin{proofpf}
     \pflongnumbers
 \step{1}{
      \assume{
        \begin{pfenum}
          \item (H1) ${\sst}.(\darkmath{\lstacomp{\act_1}{\act_2}}) \rightarrow \left\{ (\sst_i, [\sexp_{1_i}, \sexp_{2_i}])^{\result_i}|_{i \in I} \right\}$
          \item (H2) $\action{\sint(\sst)}{(\lstacomp{\act_1}{\act_2})}{[\sintp(\sexp_1), \sintp(\sexp_2)]}{(\st', [\vl_1', \vl_2'])^\result}$
        \end{pfenum}
      }
    }
      \step{2}{[From Action Composition on H1 and the assumption that $\alpha_1$ always succeeds]}
        \begin{pfenum*}
           \item $\sst.\darkmath{\act_1}\!(\sexp) \rightarrow \left\{ (\sst_k, \sexp_{1_k})^{\osucc}|_{k \in K} \right\}$  \label{1m1}
           \item $(\sst_i.\darkmath{\act_2}\!(\sexp_2) \rightarrow \left\{ (\sst_{kj}, \sexp_{2_{kj}})^{\result_{kj}}|_{j \in J_k} \right\})|_{k \in K}$ \label{1m2}
           \item $\left\{ (\sst_i, [\sexp_{1_i}, \sexp_{2_i}])^{\result_i}|_{i \in I} \right\} = \bigcup_{k \in K} \bigcup_{j \in J_k}(\sst_{kj}, [\sexp_{1_k}, \sexp_{2_{kj}}])^{\result_{kj}}$
      \end{pfenum*}
      \step{3}{[From Action Composition on H2 and the assumption that $\alpha_1$ always succeeds]}
        \begin{pfenum*}
           \item $\action{\sint(\sst)}{\act_1)}{\sintp(\sexp_1)}{(\st'', \vl_1')^\osucc}$
           \item $\action{\st''}{(\act_2)}{\sintp(\sexp_2)}{(\st', \vl_2')^\result}$ 
      \end{pfenum*}

      \step{4}{Obtain $\varepsilon'' \geq \varepsilon$, $k \in K$: $(\st'', \vl_1')^\osucc \in \sintpe{\varepsilon''}\!\left(\left\{ (\sst_k, \sexp_{1_k})^{\osucc}|_{k \in K} \right\}\right) \gand \st'' = \sintpe{\varepsilon''}(\sst_k) \gand \vl_1' = \sintpe{\varepsilon''}(\sexp_{1_k})$ [by [SA-FS] applied to 2.1 and 3.1]}
      \step{6a}{$\sintpe{\varepsilon''}(\sexp_2) = \sintp(\sexp_2)$ [since $\varepsilon'' \geq \varepsilon$]}
      \step{6b}{Obtain $\varepsilon' \geq \varepsilon''$, $j \in J_k$: $(\st', \vl_2')^\result \in \left\{ (\sst_{kj}, \sexp_{2_{kj}})^{\result_{kj}}|_{j \in J_k} \right\} \gand \st' = \sintpe{\varepsilon'}(\sst_{kj}) \gand \vl_2' = \sintpe{\varepsilon'}(\sexp_{2_{kj}})$ [by [SA-FS] applied to 2.2 and 3.2, given 4 and 5]}
            \step{6c}{$\sintpe{\varepsilon'}(\sexp_1) = \sintp(\sexp_1)$ [since $\varepsilon' \geq \varepsilon'' \geq \varepsilon$]}
      \step{7}{(G1) [by 6, 7, and 2.3]}
          \end{proofpf} 
\end{proof}

We can now move on to \gil execution:

\begin{lemma}[Forward Soundness: \gil One-Step]\label{lem:gilfsone}
Given a concrete and a symbolic state model, $\gstate(\actions, \cmemory(\actions)) \defeq \langle\sset{\gstate}, \kwT{ea} \rangle$, and  $\hat\gstate(\actions, \smemory(\actions)) \defeq \langle\sset{\hat\gstate}, \kwT{ea} \rangle$, respectively,
if the memory actions satisfy Requirements~\ref{req:mem:fs}, then the corresponding \gil one-step symbolic semantics satisfies the following:
\begin{mathpar}
\inferrule[GIL: Forward Soundness: Interpretations (GIL-FS-CFG-I)]{}{
          {\begin{array}{l}
\prog \vdash \scf \rightarrow_c \ssscf'  \gand  
\prog \vdash  \sintp(\scf)  \semarrow_c {\cf'}  \implies \exists \varepsilon' \geq \varepsilon, \scf' \in \ssscf'.~\cf' = \sintpe{\varepsilon'}(\scf')
 \end{array}}}
\end{mathpar}
\end{lemma}

\begin{proof}
The cases which involve only one state action execution or only one action composition follow trivially from Lemma~\ref{req:st:fs} and Lemma~\ref{lem:sactfs}. The only two remaining cases are the function call rule and the return rule. We prove the function call rule; the return rule is proven analogously: 
\case{\prooflab{Call}}
\begin{proofpf}
  \pflongnumbers
  \step{s1} {Hypotheses:
    \begin{pfenum}
      \item[(H1)] $\prog \vdash \scf \rightarrow \ssscf'$
      \item[(H2)] $\prog \vdash  \sintp(\scf)  \semarrow {\cf'}$
    \end{pfenum}}
  \step{s2} {[by H1, \prooflab{Call} case]
    \begin{pfenum}
  	  \item $\hat\cf = \tup{\sst, \hat\cs, i}$
  	  \item $\cmd(\prog, \hat\cs, i) = \x := \e(\e')$
	  \item $\uaction{\sst}{(\lstacomp{\kwT{eval}_{\e}}{(\lstacomp{\kwT{eval}_{\e'}}{\kwT{getStore}})})}{-}{\{ (\sst_1, [\f, [\sexp, \ssto']])^\osucc\}}$
	  \item $\uaction{\sst_1}{\kwT{setStore}}{\btup{[\f.\kwT{arg}, \sexp]}}{\{(\sst', \true)^\osucc\}}$
	  \item $\hat\cf' = \tup{\sst', \tup{\f, \x, \ssto', i{+}1} \cons \hat\cs, 0}^\cont$
	  \item $\ssscf' = \{ \hat\cf' \}$ [by the determinism of all of the used state actions]
    \end{pfenum} }
  \step{s2} {[by H2, \prooflab{Call} case]
    \begin{pfenum}
  	  \item $\sintp(\cf) = \tup{\sintp(\sst), \sintp(\hat\cs), i}$
  	  \item $\cmd(\prog, \sintp(\hat\cs), i) = \x := \e(\e')$
	  \item $\action{\sintp(\sst)}{(\lstacomp{\kwT{eval}_{\e}}{(\lstacomp{\kwT{eval}_{\e'}}{\kwT{getStore}})})}{-}{ (\st_1, [\f', [\vl, \sto']])^\osucc}$
	  \item $\uaction{\st_1}{\kwT{setStore}}{\btup{[\f.\kwT{arg}, \vl]}}{\{(\st', \true)^\osucc\}}$
	  \item $\cf' = \tup{\st', \tup{\f, \x, \sto', i{+}1} \cons \sintp(\hat\cs), 0}^\cont$
    \end{pfenum} }
  \step{s5} {$\exists \varepsilon'' \geq \varepsilon.~\st_1 = \sintpe{\varepsilon''}({\sst_1}) \gand \f' = \f \gand \vl = \sintpe{\varepsilon''}(\sexp) \gand \sto' = \sintpe{\varepsilon''}(\ssto')$ [by \prooflab{SA-FS} +  Lemma~\ref{lem:sactfs} + 2.3 + 3.3]}
  \step{xx} {$\exists \varepsilon' \geq \varepsilon''.~\st' = \sintpe{\varepsilon'}(\sst')$
   [by \prooflab{SA-FS}, given \stepref{s2}.4 + 3.4]}
            \step{6c}{$\sintpe{\varepsilon'}(\ssto') = \sintpe{\varepsilon''}(\ssto') \gand \sintpe{\varepsilon'}(\ssto') = \sintpe{\varepsilon}(\hat\cs)$ [since $\varepsilon' \geq \varepsilon'' \geq \varepsilon$]}
  \step{s8} {$\cf' = \sintpe{\varepsilon'}(\scf')$ [by 2.5, 3.5, 4, 5, and 6]}
  \step{s9} {(G1) [by 1.6 and 7, taking $\varepsilon' = \varepsilon'$ and $\scf' = \scf'$]}
\end{proofpf}
\end{proof}

We next rephrase the last property in terms of models and extend it to the one-step collecting semantics. In order to do this, we extend the notion of models to sets of configurations in the standard way:
$$
\modls(\ssscf) \defeq \{ \sintp(\scf) \mid \scf \in \ssscf, \Vint\sintp\scf \}
$$

\begin{lemma}[\gil One-Step Collecting Semantics: Forward Soundness] 
\label{lem:goscsfs:mod}
The following soundness property holds for the \gil one-step collecting semantics:
\begin{mathpar}
\inferrule[GIL: Forward Soundness: Config (GIL-FS-CFG)]{}{
          {\begin{array}{l}
\prog \vdash \scf \rightarrow_c \ssscf' \gand \cf \in \modls({\scf})\gand 
\prog \vdash  \cf  \semarrow_c {\cf'}  \implies \exists \scf' \in \ssscf'.~\cf' \in
 \modls(\scf') 
 \end{array}}}
\end{mathpar}
\begin{mathpar}
\inferrule[GIL: Forward Soundness: Collecting One-Step (GIL-FS-COS)]{}{
          {\begin{array}{l}
\prog \vdash \ssscf \rightarrow \ssscf'  \gand \scf \in \ssscf \gand \cf \in \modls({\scf})\gand 
 \prog \vdash  \cf  \semarrow_c {\cf'}  \\
 \qquad \implies \exists \scf' \in \ssscf'.~\cf' \in
 \modls(\scf') \gand \prog \vdash  \scf  \semarrow_c {\scf'}  \\[1mm]
\prog \vdash \ssscf \rightarrow \ssscf'  \gand \cf \in \modls({\ssscf})\gand 
 \isfinal(\cf)  \implies \cf \in
 \modls(\ssscf') 
 \end{array}}}
\end{mathpar}
\end{lemma}
\begin{proof}
The first part follows straightforwardly from the statement of Lemma~\ref{lem:gilfsone}, given that, by definition of models, $\sintpe{\varepsilon}(\scf) \in \modls(\scf)$ and $\sintpe{\varepsilon'}(\scf') \in \modls(\scf')$. The second part follows directly from the definition of the one-step collecting semantics.
\end{proof}

\begin{lemma}[\gil Multi-Step Collecting Semantics: Forward Soundness] 
\label{lem:gmscsfs:mod}
\begin{mathpar}
\inferrule[GIL: Forward Soundness: Collecting Multi-Step (GIL-FS-CMS)]{}{
          {\begin{array}{l}
\sinv(\ssscf) \gand \prog \vdash \ssscf \rightarrow^n \ssscf'  \gand \scf \in \ssscf \gand \cf \in \modls({\scf})~\gand \\ 
\quad \left( \prog \vdash  \cf  \semarrow^n_C {\cf'} \lor \left(m < n \gand \prog \vdash  \cf  \semarrow^m_C {\cf'} \gand \isfinal(\cf') \right) \right) \\
\qquad \implies \exists \scf' \in \ssscf'.~\cf' \in
 \modls(\scf') \gand \left( \prog \vdash  \scf  \semarrow^n_C {\scf'} \lor \left(m < n \gand \prog \vdash  \scf  \semarrow^m_C {\scf'} \gand \isfinal(\scf') \right) \right) 
 \end{array}}}
\end{mathpar}
\end{lemma}

\begin{proof}
By induction on $n$. The base case, when $n = 0$, is trivial. We prove the inductive case, for $n = k + 1$, as follows:
\begin{proofpf}
  \pflongnumbers
    \step{2ass}{
      \begin{pfenum}
        \item[(H1)] $\prog \vdash \ssscf \rightarrow^{k+1} \ssscf'$
        \item[(H2)] $\cf \in \modls({\ssscf})$
        \item[(H3)] $\left( \prog \vdash  \cf  \semarrow^{k+1}_C {\cf'} \lor \left(m < {k+1} \gand \prog \vdash  \cf  \semarrow^m_C {\cf'} \gand \isfinal(\cf') \right) \right)$
      \end{pfenum}
    }
		\step{2a}{$\exists \ssscf''.$}
		\begin{pfenum*}
			\item $\prog \vdash \ssscf \rightarrow \ssscf''$ [From H1]  \label{2a1}
			\item $\prog \vdash \ssscf'' \rightarrow^k \ssscf'$ \label{2a2} [From H1]
		\end{pfenum*}
		\step{xb}{BRANCH: $\isfinal(\cf)$}
		\begin{pfenum*}
		\item {$\cf \in \modls(\ssscf'')$ [by 3 and Lemma~\ref{lem:goscsfs:mod}]}
		\item {$\prog \vdash  \cf  \semarrow^0 {\cf} \gand \isfinal(\cf) $ [by H3 and 3]}
		\item {$\cf \in \modls(\ssscf')$ [by IH, 2.2, 3.1, and 3.2]}
		\end{pfenum*}
		\step{xc}{BRANCH: $\lnot\isfinal(\cf): \exists \cf''.$}
		\begin{pfenum*}
		\item $\prog \vdash  \cf  \semarrow_c \cf''$ [by H3 and 4]
		\item $\prog \vdash  \cf''  \semarrow^{k}_{C'} {\cf'}$ [by H3 and 4]
		\item {$\exists \scf''.~\scf'' \in \modls(\ssscf'') \gand \cf'' \in \modls(\scf'') \land \prog \vdash \scf \semarrow_c \scf''$ [by 4.1 and Lemma~\ref{lem:goscsfs:mod}]}
		\item $\left( \prog \vdash  \cf''  \semarrow^{k}_{C'} {\cf'} \lor \left(m < {k} \gand \prog \vdash  \cf''  \semarrow^m_{C'} {\cf'} \gand \isfinal(\cf') \right) \right)$ [by 4.2]
		\item {$\exists \scf' \in \ssscf'.~\cf' \in
 \modls(\scf') \gand  \prog \vdash  \scf''  \semarrow^k_{C'} {\scf'}$ [by IH, 2.2, 4.2, and 4.3]}
 		\item $\cf' \in
 \modls(\scf') \gand \left( \prog \vdash  \scf  \semarrow^n_C {\scf'} \lor \left(m < n \gand \prog \vdash  \scf  \semarrow^m_C {\scf'} \gand \isfinal(\scf') \right) \right)$ (G1) [by 4.1 and 4.5] 
		\end{pfenum*}
	\end{proofpf}
\end{proof}

Finally, we extend the forward soundness property to the to-termination \gil collecting semantics, which follows straightforwardly from its definition and Lemma~\ref{lem:gmscsfs:mod}, and gives the partial verification guarantee for the \gil semantics: 

\begin{theorem}[\gil Forward Soundness] 
\label{lem:gttcsfs:mod}
\gil execution satisfies the following property:
\begin{mathpar}
\inferrule[GIL: Forward Soundness]{}{
          {\begin{array}{l}
\prog \vdash \{ \scf \} \csemarrowfin \ssscf' \gand \cf \in \modls({\scf}) \gand~\prog \vdash  \cf  \csemarrowfin_C {\cf'} \\ \qquad \implies \exists \scf'.~\scf' \in \ssscf' \gand \cf' \in \modls(\scf') \gand \prog \vdash \scf \csemarrowfin_C \scf'
 \end{array}}}
\end{mathpar}
\end{theorem}

This property captures partial verification of the \gil symbolic execution: it states that the set of concrete configurations reachable from $\modls({\hat\cf})$ by executing $C$ is \emph{over-approximated} by $\modls(\ssscf')$, meaning that any final concrete configuration reachable from any concrete state in $\modls({\hat\cf})$ by executing $C$ has to be in $\modls(\ssscf')$.

%% file: sections/spec.tex

\section{\gillian Compositional Execution}
\label{sec:gilspec}

We introduce \emph{compositional memory models}, which are suitable for \emph{compositional} symbolic execution. They are defined in terms of \emph{core predicates} and their \emph{consumers} and \emph{producers}. Core predicates represent separation-logic assertions that describe the fundamental units of the target-language memory models. a consumer frames off the resource of a given core predicate, and a producer frames that resource on. We demonstrate how compositional memory models induce execution memory models. In addition, we show that all of correctness results of execution memory models, and prove a frame preservation property for \gil execution, essential for compositionality.

\subsection{Compositional Memory Models} 
\label{sec:cmms}

Compositional memory models are parametric on a value set $\gval$ and a core predicate set $\corePreds \ni \corePred$, where every core predicate $\corePred$ is associated with a list of in-parameters, $\gv_i$, and a list of out-parameters, $\gv_o$, giving rise to the {\em core
  predicate resource} (CPR), $\cPred{\gv_i \cdot \gv_o}$.
The core predicate set induces a predicate action set,
$\actions_\corePreds \defeq \bigcup_{\corePred \in \corePreds} \{
\fgetter{\corePred}, \fsetter{\corePred} \}$, comprising
\emph{consumer} and \emph{producer} actions for each core predicate,
and a memory action set, comprising \emph{getter} and \emph{setter}
actions, $\actions \defeq \bigcup_{\corePred \in \corePreds} \{
\fgetterx{\corePred}, \fsetterx{\corePred} \}$,  for each
core predicate.

\begin{definition}[Compositional Memory Model]\label{def:smm}
Given a value set $\gval$ and a core predicate set $\corePreds$, a 
compositional memory model, $\cmemory \defeq \tup{\sset{\cmemory}, \sinv, 
\cea_{\!\corePreds}}$  comprises:
\begin{enumerate}[leftmargin=*,label=\arabic*)]
\item a partial commutative monoid (PCM), $\sset{\cmemory} = (\sset\cmemory, \stcomp,
  \stzero)$, with a carrier set of \myuline{memories},  $\sset{\cmemory} \ni \cmem$, 
a partial, associative and
    commutative binary composition operator,~$\stcomp$, and the unit 
    element, $\stzero$, denoting the empty memory,  which is indivisible:
    that is, $\pc \vdash \mu \stcomp - = 0 \Rightarrow \pc \vdash \mu  = 0$;

\item a \myuline{well-formedness relation}, $\sinv \subseteq
    \sset{\cmemory} \times \pcs$, with  $\sinv_{\!\pc}(\cmem)$ denoting 
    that memory~$\cmem$ is well-formed in context $\pc$, which, in addition to the properties 2.1-2.2 of Definition~\ref{def:emm}, has the following properties, which describe its relationship to the memory PCM:
    \begin{enumerate}[left=0cm,label=\arabic{enumi}.\arabic*)]
    \setcounter{enumii}{2}
    \item well-formedness is compatible with memory
  composition: $$\sinv_{\!\pc}(\cmem_1 \stcomp \cmem_2) \implies
  \sinv_{\!\pc}(\cmem_1) \land \sinv_{\!\pc}(\cmem_2);$$
\item the empty memory,  $\stzero$, is   well-formed in any
  satisfiable context:  $\pi \;\sat \implies \sinv_{\!\pc}(\stzero)$; 

  \end{enumerate}

\item 
a \myuline{predicate action execution function}, $\cea_\corePreds : \actions_\corePreds \times \sset{\cmemory} \times \gval \times \pcs \pmap \power{\sset{\cmemory} \times \gval \times \results \times \pcs}$,
  where $\cea_\corePreds  (\mu, \alpha, \gv, \pc) =  \{ (\cmem_i, \gv_i , \result_i
  ,\pc_i)     : i  \in I \}$, for countable $I$, with notation 
  $\fcolactionx \cmem \act \gv {\pc} {\left\{ (\cmem_i,
      \gv_i)_{\pc_i}^{\result_i}|_{i \in I} \right\}}$ for all
  outcomes and  $\faction{\cmem}{\act}{\gv}{}{\cmem_i}{\gv_i}{\pc_i }
{\result_i}$ for a specific outcome, satisfies, in addition to the properties 3.1-3.6 of Definition~\ref{def:emm}, the following properties:
\begin{enumerate}[left=0cm,label=\arabic{enumi}.\arabic*)]
\setcounter{enumii}{6}
\item the producers succeed without branching in the empty memory in any satisfiable context, producing the CPR:
$$
\begin{array}{r@{.~}l}
\pc~\sat \implies \faction{\stzero}{\fsetter{\corePred}}{\gv_i \cdot \gv_o}{\pc}{{\cPred{\gv_i \cdot \gv_o}}}{\true}{\pc}{\osucc} 
\end{array}
$$
\item if the producers succeed, they add on the CPR:
$$
\begin{array}{c}
\faction \cmem {\fsetter{\corePred}} {\gv_i \cdot \gv_o} {\pc} {\cmem'} {\true} {\pc'} \osucc \implies \cmem' = \cmem \stcomp \cPred{\gv_i \cdot \gv_o} 
\end{array}
$$
\item the producers cannot return a missing information error:
$$
\faction \cmem {\fsetter{\corePred}} {\gv_i \cdot \gv_o} {\pc} {\cmem'} {\gv'} {\pc'} \result \implies \result \neq \omiss
$$

\item the consumers succeed if and only if the CPR is present in memory:
$$
\begin{array}{r@{~}c@{~}l}
\faction \cmem {\fgetter{\corePred}} {\gv_i} {\pc} {\cmem'} {\gv_o} {\pc'} \osucc & \implies & \pc' \Leftrightarrow \pc \land (\cmem = \cmem' \stcomp \cPred{\gv_i \cdot \gv_o)} \\[1mm]
\sinv_{\!\pc}(\cmem) \gand \pc' \Leftrightarrow (\cmem = \cmem' \stcomp \cPred{\gv_i \cdot \gv_o}) \land (\pc \land \pc'~\sat) & \implies & \faction \cmem {\fgetter{\corePred}} {\gv_i} {\pc} {\cmem'} {\gv_o} {\pc \land \pc'} \osucc 
\end{array}
$$

\item 
consumed CPR can always be successfully re-produced:
$$
\begin{array}{l}
\faction{\cmem}{\fgetter{\corePred}}{\gv_i}{\pc}{\cmem'}{\gv_o}{\pc'}{\osucc} \implies~\faction{\cmem'}{\fsetter{\corePred}}{\gv_i \cdot \gv_o'} {\pc'} {\cmem''} \true {\pc'} {\osucc}  
\end{array}
$$
\item producers fail in a memory in which
  consumers succeed, and succeed in a memory in which consumers
  return a missing information error (and vice versa for the latter):
$$
\begin{array}{r@{~}c@{~}l}
\faction \cmem {\fgetter{\corePred}} {\gv_i} {\pc} {\cmem'} {\gv_o} {\pc'} \osucc & \implies & \faction \cmem {\fsetter{\corePred}} {\gv_i \cdot -} {\pc} {\cmem} {\false} {\pc'} \oerr \\[1mm]
\faction \cmem {\fgetter{\corePred}} {\gv_i} {\pc} {\cmem} {\false} {\pc'} \omiss & \iff & \faction \cmem {\fsetter{\corePred}} {\gv_i \cdot \gv_o} {\pc} {\cmem \stcomp \cPred{\gv_i \cdot \gv_o}} {\true} {\pc'} \osucc 
\end{array}
$$

\item erroneous executions of the consumers and producers are frame-preserving: for all $\act \in \actions_\corePreds$:
$$
{\begin{array}{l}
\faction \cmem {\act} \gv {\pc} {\cmem'} {\gv'} {\pc'} \oerr 
 \gand \sinv_{\!\pc \land \pc_f}(\cmem \stcomp \cmem_{f}) \gand (\pc' \land \pc_f~\sat)
\\ \qquad
 \implies 
\faction {(\cmem \stcomp \cmem_f)} {\act} \gv {\pc \land \pc_f} {{\cmem'} \stcomp \cmem_f} {\gv'} {\pc' \land \pc_f} \oerr
\end{array}}
$$
\item successful executions of the producers are frame-cancelling:
$$
\begin{array}{l}
\faction {(\cmem \stcomp \cmem_f)} {\fsetter{\corePred}} {\gv_i \cdot \gv_o} {\pc} {\cmem \stcomp \cmem_f \stcomp \cPred{\gv_i \cdot \gv_o}} {\gv'} {\pc'} \result \\[1mm] \qquad\qquad \implies \faction {\cmem} {\fsetter{\corePred}} {\gv_i \cdot \gv_o} {\pc'} {\cmem \stcomp \cPred{\gv_i \cdot \gv_o}} {\gv'} {\pc'} \osucc 
\end{array}
$$

\end{enumerate}
\end{enumerate}
\end{definition}

\begin{definition}[Concrete and Symbolic Compositional Memory Model]\label{def:cssmm}
A {\em concrete} compositional memory model, 
$\cmemory \defeq \tup{\sset{\cmemory}, \sinv, 
\cea_{\!\corePreds}}$, is a compositional memory model with 
value set given by the \gil value set, $\vals$. 
A {\em symbolic} compositional memory model, $\smemory= \tup{\sset{\smemory}, \sinv, 
\sea_{\!\corePreds}}$, is a compositional memory model with value set given by the  set of \gil symbolic expressions, $\sexps$.
\end{definition}

\myparagraph{Discussion and Observations} 
The PCM requirement is standard for compositional reasoning, and the two new well-formedness properties ensure compatibility of well-formedness with the PCM. Property 3.7 pinpoints the CPR, and properties 3.8 and 3.10 capture its relationship with the consumers and producers. Note that it is not possible to characterise erroneous behaviours in terms of CPR presence of absence, due to TL-specific errors, such as the C access out of bounds. Instead, we require properties 3.11-3.12, with the former being essential later on for proving frame preservation of getters and setters, as well as for  specification use. As for the frame preservation requirement, one can notice that it is incompatible with resource allocation (as the allocated resource could be part of the frame, causing resource duplication in $\cmem' \stcomp \cmem_f$). This will be further discussed in the context of the frame preservation of the \gil semantics. Note that the requirement $\pc' \land \pc_f~\sat$ means that the branch that the execution took (from $\pc$ to $\pc'$) must remain feasible with the added frame.

\myparagraph{Getters and Setters} The getter and setter memory actions are defined in terms of consumers and producers as follows:
\begin{mathpar}
\inferrule[\textsc{Getter: Success}]
  {
    \faction{\cmem}{\fgetter{\corePred}}{\gv_i} \pc {\cmem'}{\gv_o}{\pc'}{\osucc}  \\\\
    \faction{\cmem'}{\fsetter{\corePred}}{\gv_i \cdot \gv_o} {\pc'} {\cmem''}{\true}{\pc'}{\osucc} 
  }
  {\faction{\cmem}{\fgetterx{\corePred}}{\gv_i} \pc {\cmem''}{\gv_o}{\pc'} {\osucc} }  
 \and
 \inferrule[\textsc{Setter: Success}]
  {
    \faction{\cmem}{\fgetter{\corePred}}{\gv_i}\pc{\cmem'}{-}{\pc'}{\osucc} \\\\
    \faction{\cmem'}{\fsetter{\corePred}}{\gv_i \cdot \gv_o}{\pc'}{\cmem''}{\true}{\pc'}{\osucc} 
  }
    {\faction{\cmem}{\fsetterx{\corePred}}{\gv_i \cdot \gv_o}\pc {\cmem''}{\true}{\pc'}{\osucc} }  
\end{mathpar}
\begin{mathpar}
  \inferrule[\textsc{Getter: Non-Success}]
  {
    \faction{\cmem}{\fgetter{\corePred}}{\gv_i} \pc {\cmem}{\false}{\pc'} {\result} \and \result \neq \osucc
  }
  {\faction{\cmem}{\fgetterx{\corePred}}{\gv_i}\pc{\cmem}{\false}{\pc'}{\result} }  
 \and
 \inferrule[\textsc{Setter: Non-Success}]
  {
    \faction{\cmem}{\fgetter{\corePred}}{\gv_i} \pc {\cmem}{\false}{\pc'}{\result} \and \result \neq \osucc
  }
    {\faction{\cmem}{\fsetterx{\corePred}}{\gv_i \cdot \gv_o} \pc {\cmem}{\false}{\pc'}{\result} }  
\end{mathpar}
and correspond to CPR lookup and mutation, respectively. We can trivially show that, if the consumers and producers satisfy the properties 3.1-3.13 of Definition~\ref{def:smm}, then the getters and setters satisfy properties 3.1-3.6, meaning that they can be used in whole-program \gil execution.

\subsection{Compositional State Models}
\label{subsec:csms}

Compositional state models are defined analogously to execution state models. Recall the set of  basic state actions,  $\sactions
= \{ \kwT{setVar}_{\x} \}_{\x \in \xs} \cup \set{\kwT{setStore}, \kwT{getStore}} 
\cup \{ \kwT{eval}_{\e} \}_{\e \in \exprs} \cup \{\kwT{assume},$ $\kwT{uSym}, \kwT{iSym}\}
$, and recall that a set of core predicates $\corePreds$ induces a set of consumers and producers, $\actions_\corePreds$, and a set 
of getters and setters, $\actions$.

\begin{definition}[Compositional State Model]
\label{def:csm}

  Given a value set $\gval$, a set of interpreted symbols $\mathsf{I}
\subset \gval$,
a set of core predicates $\corePreds$,  a compositional memory model $\cmemory
= \tup{\sset{\cmemory}, \sinv,$  $\cea}$, and an allocator
$\allocator = \tup{\sset{\allocator}, \allocf}$,
a {\em compositional state model} is defined by 
$\gstate \defeq
\langle(\sset{\gstate}, \stcomp, \stzero), \ceval{~}{}, \kwT{ea} \rangle$, comprising:
\begin{itemize}[leftmargin=*]
\item a set of \myuline{states}, $\sset{\gstate} \subseteq 
  \sset{\cmemory} \times (\xs \pmap \gval) \times \sset\allocator
  \times \pcs$,  with  $\st = \tup{\cmem,
    \sto, \arec, \pc} \in \sset{\gstate}  $ { if and only if $\sinv_{\!\pc}(\cmem)$}, with state composition $\stcomp$ and the neutral element $\stzero$ defined as
  $$
 \st_1 \stcomp \st_2 = \begin{cases}
    \tup{\cmem_1 \stcomp \cmem_2, \sto_1, \arec_1 \stcomp \arec_2, \pc_1 \land \pc_2}, &
    \st_1 = \tup{\cmem_1, \sto_1, \arec_1, \pc_1}, \st_2 =
    \tup{\cmem_2, \emptyset, \arec_2, \pc_2}, \\
    &  \sinv_{\!\pc_1 \land \pc_2}(\cmem_1 \stcomp \cmem_2)~\text{(and vice versa with respect to stores)} \\
\text{undefined}, & \text{otherwise} 
\end{cases}
$$
       and $\stzero = \tup{\stzero, \emptyset, \emptyset, \true}$   , where allocator composition is defined as
   \begin{mathpar}
    \inferrule[]
{(\arec' \stcomp \arec'')(Y) = 
{\begin{cases} 
	\arec_{o}'(Y) \cup \arec_{o}''(Y), & \text{if $Y \in \domain(\arec_1)$ or $Y \in \domain(\arec_2)$} \\
	\text{undefined}, & \text{otherwise} 
\end{cases}}}{}
\end{mathpar}
\item an {\em expression evaluation function} $\ceval{~}{} : \exprs
  \times (\xs \pmap \gval) \tmap \gval$, where $\ceval{\e}{\sto} \in
  \gval$ denotes the evaluation of a \gil expression $e$ with respect
  to variable store $\sto$;

\item a \myuline{state action execution function}, $\kwT{ea} :
  \sset{\gstate} \times (\mactions_\corePreds \uplus A \uplus \sactions) \times \gval
  \pmap \power{\sset{\gstate} \times \gval \times \results}$,
  defined by:
  for $\st = \tup{\cmem,
    \sto, \arec, \pc}$ and $\alpha \in A_\corePreds \uplus A$,
  
\medskip
\noindent
\begin{tabular}{l@{~}c@{~}l}
     $\kwT{ea}(\st, \act, \gv)$
    & $\semeq$ & 
    $\{ (\tup{\cmem', \sto, \arec, \pc'}, \gv')^\result \mid  \faction \cmem {\act} \gv {\pc} {\cmem'} {\gv'} {\pc'} \result 
\} $ 
\\
  $\kwT{ea}(\st, \kwT{setVar}_{\x},  \gv)$ 
       & $\semeq$ & 
       $\set{ (\tup{\cmem, \sto[\x \mapsto \gv], \arec, \pc}, \true)^\osucc}$ \\ 
  $\kwT{ea}(\st, \kwT{setStore}, \sto')$
     & $\semeq$ & 
     $\set{(\tup{\cmem, \sto', \arec, \pc}, \true)^\osucc}$ \\
  $\kwT{ea}(\st, \kwT{getStore}, -)$
    & $\semeq$ & 
     $\set{ (\st, \sto)^\osucc}$ \\
  $\kwT{ea}(\st, \kwT{eval}_{\e}, -)$
    & $\semeq$ & 
     $\set{ (\st, \ceval{\e}{\sto})^\osucc}$ \\
  $\kwT{ea}(\st, \kwT{assume}, \pc')$
    & $\semeq$ & 
    $\lbrace (\tup{\cmem, \sto, \arec, \pc \land \pc'}, \true)^\osucc \mid \pc \land \pc'~\sat) \rbrace$ \\
    $\kwT{ea}(\st, \kwT{uSym}, n)$
    & $\semeq$ &
     $\lbrace (\tup{\cmem, \sto', \arec', \pc}, \lst \loc)^\osucc \mid \alloc{\arec}{n}{\arec', \lst \loc}{\clocs} \rbrace $ \\
    $\kwT{ea}(\st, \kwT{iSym}, n)$
    & $\semeq$ &
     $\lbrace (\tup{\cmem, \sto', \arec, \pc}, \lst \gv)^\osucc \mid
                 \alloc{\arec}{n}{\arec', \lst \gv}{\mathsf{I}}
                 \rbrace $ 
              
\end{tabular}
\end{itemize}
\end{definition}

\begin{definition}[Concrete and Symbolic Compositional State
  Model]\label{def:csesm}
  Concrete compositional state models are obtained by having the value set and the set of interpreted symbols be the set of \gil values, $\vals$, and symbolic ones are obtained by having the value set be the set of \gil symbolic expressions,~$\sexps$, and the set of interpreted symbolic be the set of symbolic variables, $\lxs$. The expression evaluation function is defined in the standard way for both concrete and symbolic~ execution.
\end{definition}

\myparagraph{Discussion and Observations}
There are other  ways in which state composition could have been
defined. Our choice was enough for our use case, as our compositionality boundary is at the level of procedures, and procedures
specifications do not talk about program variables.

\subsection{Correctness: Forward Soundness and Backward Completeness}
\label{sec:sublift}

The forward soundness and backward completeness for compositional memory and state models follow  from the corresponding correctness results for their execution counterparts. In particular, only the following proposition is required.

\begin{propos}\label{cor:spec:simpl} 
Given a concrete and a symbolic compositional memory model, $\cmemory(\vals, \corePreds)$ and $\smemory(\sexps, \corePreds)$,
if the consumers and producers induced by $\corePreds$ satisfy Requirements~\ref{des:memact}/Requirements~\ref{req:mem:fs}, then the  induced getters and setters satisfy those requirements as well.
\end{propos}

\begin{proof}
The proof for backward completeness is trivial. The proof for forward soundness uses Proposition~\ref{prop:intermezzo} and backward completeness (see proof of Proposition~\ref{prodcons:mass} for more details).
\end{proof}

With this proposition in place, the remaining properties of the state actions and the \gil execution follow directly, as all of the actions of the compositional memory model (the consumers, producers, getters, and setters) meet the appropriate requirements.

\subsection{Correctness: Frame Preservation}
\label{cor:fp}

\subsubsection{Compositional Memory Models}
The following properties of consumers and producers can be derived from the properties given in Definition~\ref{def:smm}.

\begin{propos}
\label{prop:succexc}
Successful executions of consumers are frame-preserving.
\end{propos}
\begin{proof}
Assume  \hyp 1~$\faction \cmem {\fgetter{\corePred}} {\gv_i} {\pc} {\cmem'} {\gv_o} {\pc'} \osucc$, \hyp 2~$\sinv_{\!\pc \land \pc_f}(\cmem \stcomp \cmem_{f})$, and \hyp 3~$\pc' \land \pc_f~\sat$.
Then, from H1 and 3.10, we have \hyp 4~$\pc' \Leftrightarrow \pc \land (\cmem = \cmem' \stcomp \cPred{\gv_i \cdot \gv_o})$.
From H3 and H4, we have \hyp 5~$\pc \land \pc_f \land (\cmem = \cmem' \stcomp \cPred{\gv_i \cdot \gv_o})~\sat$. 
From H5 and properties of $\stcomp$, we have \hyp 6~$\pc \land \pc_f \land (\cmem \stcomp \cmem_f = (\cmem' \stcomp \cmem_f) \stcomp \cPred{\gv_i \cdot \gv_o})~\sat$. 
Finally, from H2, H6, and property 3.10, we obtain the desired $\faction {(\cmem \stcomp \cmem_f)} {\fgetter{\corePred}} {\gv_i} {\pc \land \pc_f} {{\cmem'} \stcomp \cmem_f} {\gv_o} {\pc' \land \pc_f} \osucc$.
\end{proof}

\begin{propos}[Producer-Consumer Invertibility]
\label{prop:cprinv}
$$
\begin{array}{l}
\faction{\cmem}{\fsetter{\corePred}}{\gv_i \cdot \gv_o}{\pc}{\cmem'}{\true}{\pc'}{\osucc} \implies 
\faction{\cmem'}{\fgetter{\corePred}}{\gv_i} {\pc'} \cmem {\gv_o} {\pc'} {\osucc}
\end{array}
$$
\end{propos}
\begin{proof}
Assume \hyp 1 $\faction{\cmem}{\fsetter{\corePred}}{\gv_i \cdot \gv_o}{\pc}{\cmem'}{\true}{\pc'}{\osucc}$. Then, by 3.8, we have that \hyp 2 $\cmem' = \cmem \stcomp \cPred{\gv_i \cdot \gv_o}$, and by 3.1, we have that \hyp 3 $\sinv_{\!\pc'}(\cmem')$. From H2, we have that \hyp 4~$\pc' \vdash \cmem' = \cmem \stcomp \cPred{\gv_i \cdot \gv_o}$, and from 3.10, H4 and H3, we obtain $\faction{\cmem'}{\fgetter{\corePred}}{\gv_i} {\pc'} \cmem {\gv_o} {\pc'} {\osucc}$.
\end{proof}

\begin{propos}[CPR Non-Duplicability]
\label{prop:cprnon}
$$
\begin{array}{l}
\faction \cmem {\fgetter{\corePred}} {\gv_i} {\pc} {\cmem'} {\gv_o} {\pc'} \osucc \implies  \faction {\cmem'} {\fgetter{\corePred}} {\gv_i} {\pc'} {\cmem'} {\false} {\pc'} \omiss  \\[1mm]
\faction \cmem {\fsetter{\corePred}} {\gv_i \cdot \gv_o} {\pc} {\cmem'} {\true} {\pc'} \osucc \implies \faction {\cmem'} {\fsetter{\corePred}} {\gv_i \cdot \gv_o} {\pc'} {\cmem'} {\false} {\pc'} \oerr \\[1mm]
\sinv_{\!\pc}(\cmem_1 \stcomp \cmem_2) \implies \pc \vdash \cPred{\gv_i \cdot -} \not\sqsubseteq \cmem_1 \vee \pc \vdash \cPred{\gv_i \cdot -}\not\sqsubseteq \cmem_2
\end{array}
$$
\end{propos}
\begin{proof}
For the first property, assume \hyp 1~$\faction \cmem {\fgetter{\corePred}} {\gv_i} {\pc} {\cmem'} {\gv_o} {\pc'} \osucc$. From H1, 3.11, and 3.8 we have \hyp 2~$\faction{\cmem'}{\fsetter{\corePred}}{\gv_i \cdot \gv_o'} {\pc'} {\cmem' \stcomp \cPred{\gv_i \cdot \gv_o}} \true {\pc'} {\osucc}$. From H2 and 3.9, we obtain the desired $\faction {\cmem'} {\fgetter{\corePred}} {\gv_i} {\pc'} {\cmem'} {\false} {\pc'} \omiss$.

For the second property, assume \hyp 1~$\faction \cmem {\fsetter{\corePred}} {\gv_i \cdot \gv_o} {\pc} {\cmem'} {\true} {\pc'} \osucc$. From H1 and Proposition~\ref{prop:cprinv}, we have \hyp 2~$\faction{\cmem'}{\fgetter{\corePred}}{\gv_i} {\pc'} \cmem {\gv_o} {\pc'} {\osucc}$. From H2 and 3.12, we obtain the desired $\faction {\cmem'} {\fsetter{\corePred}} {\gv_i \cdot \gv_o} {\pc'} {\cmem'} {\false} {\pc'} \oerr $.

The third property we prove by contradiction. Assume \hyp 1~$\sinv_{\pc}(\cmem)$, where $\pc \vdash \cmem = \cmem' \stcomp \cmem'' \stcomp \cPred{\gv_i \cdot \gv_o'} \stcomp \cPred{\gv_i \cdot \gv_o''}$, for some $\corePred$, $\cmem'$, $\cmem''$, $\gv_i$, $\gv_o'$, and $\gv_o''$. Then, given 3.10, we would be able to consecutively successfully apply the consumer for $\corePred$ on $\gv_i$, which is a contradiction with the second property of this proposition.
\end{proof}

\begin{propos}
\label{app:prop:comp}
The composition of a successful consumer followed by a producer for the same CPR is frame preserving.
\end{propos}
\begin{proof}
For the forward frame property, assume 
\hyp 1~$\faction{\cmem}{\fgetter{\corePred}}{\gv_i}\pc{\cmem'}{-}{\pc'}{\osucc}$, 
\hyp 2~$\faction{\cmem'}{\fsetter{\corePred}}{\gv_i \cdot \gv_o}{\pc'}{\cmem' \stcomp \cPred{\gv_i \cdot \gv_o}}{\true}{\pc'}{\osucc}$, 
\hyp 3~$\sinv_{\!\pc \land \pc_f}(\cmem \stcomp \cmem_{f})$, and
\hyp 4~$\pc' \land \pc_f~\sat$. 
By frame preservation of successful consumers applied to H1, H3 and H4, we obtain \hyp 5~$\faction{(\cmem \stcomp \cmem_f)}{\fgetter{\corePred}}{\gv_i}{\pc \land \pc_f}{\cmem' \stcomp \cmem_f}{-}{\pc' \land \pc_f}{\osucc}$. Then, from H5 and property 3.11, we have \hyp 6~$\faction{(\cmem' \stcomp \cmem_f)}{\fsetter{\corePred}}{\gv_i \cdot \gv_o} {\pc' \land \pc_f} {\cmem' \stcomp \cmem_f \stcomp \cPred{\gv_i \cdot \gv_o}} \true {\pc' \land \pc_f} {\osucc}$, which, together with H5, gives the desired property.
\end{proof}

\begin{lemma}
For a given compositional memory model, if the execution of consumers and producers satisfies the properties 3.1-3.14 of Definition~\ref{def:emm}, then the corresponding getters and setters satisfy properties 3.1-3.6 and their non-missing executions are frame-preserving. 
\end{lemma}

\begin{proof}
Properties 3.1-3.3 are proven trivially, and the frame preservation comes directly from Proposition~\ref{app:prop:comp} and property 3.11.
\end{proof}

The frame preservation property, in particular, means that getters and setters can also be used in compositional \gil execution.
Finally, it is straightforwardly shown that frame preservation holds for various forms of memory action composition:
\begin{propos}[Frame-Preservation and Composition]
\label{prop:fpcomp}
If two memory actions, $\act_1$ and $\act_2$, are frame-preserving for non-missing executions, then so is their composition, $\acomp{\act_1}{\act_2}$ defined by:
$$
\small
\begin{array}{r@{~}c@{~}l}
\faction{\cmem}{(\acomp{\act_1}{\act_2})}{\gv_1}{\pc}{\cmem'}{\gv'}{\pc'}\result & \iff &
 (\faction{\cmem}{\act_1}{\gv}\pc{\cmem'}{\gv'}{\pc'}\result \land r \neq \osucc)~\lor \\ &&
 (\exists \cmem'', \gv'', \pc'' . \, \faction{\cmem}{\act_1}{\gv}\pc{\cmem''}{\gv''}{\pc''}\osucc
   \ \wedge \ 
   \faction{\cmem''}{\act_2}{\gv_2}{\pc''}{\cmem'}{\gv'}{\pc'}\result) 
\end{array}
$$
\end{propos}

\begin{proof}
The error cases are trivial due to property 3.11. We consider the success cases only:
\begin{proofpf}
     \pflongnumbers
 \step{1}{
      \assume{
        \begin{pfenum}
          \item (H1) $\faction{\cmem}{(\acomp{\act_1}{\act_2})}{\gv_1}{\pc}{\cmem'}{\gv'}{\pc'}\result$
          \item (H2) $\sinv_{\pc \land \pc_f}(\cmem \stcomp \cmem_{f})$ 
          \item (H3) $\pc' \land \pc_f$~\sat       
         \end{pfenum}
      }
    }
      \step{2}{$ \exists \cmem'', \gv'', \pc''$ [From Action Composition on H1]}
        \begin{pfenum*}
           \item $\faction{\cmem}{\act_1}{\gv}{\pc}{\cmem''}{\gv''}{\pc''}\osucc$  \label{1m1}
           \item $ \faction{\cmem''}{\act_2}{\gv_2}{\pc''}{\cmem'}{\gv'}{\pc'}\result$ \label{1m2}
      \end{pfenum*}
       \step{3}{$\pc'' \land \pc_f~\sat$ [by H3 and monotonicity of contexts]}
      \step{4}{$\faction{(\cmem \stcomp \cmem_f)}{\act_1}{\gv}{\pc \land \pc_f}{\cmem'' \stcomp \cmem_f}{\gv_1}{\pc'' \land \pc_f}{\result}$ [by frame applied to 2.1, H2, and 3]}
            \step{3}{$\sinv_{\pc'' \land \pc_f}(\cmem'' \stcomp \cmem_f)$ [by 4]}

      \step{6}{$\faction{(\cmem'' \stcomp \cmem_f)}{\act_2}{\gv_2}{\pc'' \land \pc_f}{\cmem' \stcomp \cmem_f}{\gv'}{\pc' \land \pc_f}{\result}$ [by frame applied to 2.2, 5, and H3]}
      \step{7}{$\faction{(\cmem \stcomp \cmem_f)}{(\acomp{\act_1}{\act_2})}{\gv}{\pc \land \pc_f} {\cmem' \stcomp \cmem_f}{\gv'}{\pc' \land \pc_f}\result$ [by 4 and 6]}
          \end{proofpf} 
\end{proof}

The point of this property is that as long as the memories and the contexts are chained properly, then frame preservation holds regardless of how the passed-in values are obtained.

\subsubsection{Renamings} 
\newcommand{\Rnm}{\aleph}
\newcommand{\rnm}{\aleph}

We introduce renamings of memory locations and symbolic variables, to be used for the formulation of frame preservation of state actions and the \gil execution.

\begin{definition}
A \emph{renaming}, $\rnm : \mathcal{L} \uplus \lxs \tmap \mathcal{L} \uplus \lxs$ is a partial injective function that maps memory locations to memory locations and symbolic variables to symbolic variables.
\end{definition}

We extend the domain of renamings, similarly to interpretations, to all syntactic categories, including values, symbolic expressions, memories, allocators, states, and configurations in the standard way.
We require that the memory action reasoning is independent of the choice of memory locations or symbolic variables:
\begin{desiderata}[Memory Actions: Symbol Renaming]\label{des:memactsr}
Given a compositional concrete or symbolic memory model $\cmemory(\gval, \corePreds) = \tup{\sset{\cmemory}, \sinv, \cea_\corePreds}$ and a renaming $\rnm$, it must hold that:
\begin{mathpar}
\inferrule[]{}{
          {\begin{array}{l}
\faction {\cmem} \act {\gv} {\pc} {\cmem'} {\gv'} {\pc'} \result \gand \cmem, \gv, \pc \in \domain(\rnm) \implies \faction {\Rnm(\smem)} \act {\Rnm(\sexp)} {\Rnm{(\pc)}} {\Rnm(\smem')} {\Rnm(\sexp')} {\Rnm(\pc')} \result
 \end{array}}}
\end{mathpar}
\end{desiderata}
\noindent
from which we can prove the analogous property for state actions and the \gil execution, where all the symbols allocated in the final state must be in the domain of the renaming:
$$
\begin{array}{r@{~\implies~}l}
\st.\act(\gv) \rightsquigarrow (\st', \gv')^\result \gand \sst' \in \domain(\rnm) & \rnm(\st).\act(\rnm(\gv)) \rightsquigarrow (\rnm(\st'), \rnm(\gv'))^\result \\[1mm]
\prog \vdash \cf \rightsquigarrow^n \cf' \gand \cf'' \in \domain(\rnm) & \prog \vdash \rnm(\cf) \rightsquigarrow^n \rnm(\cf')
\end{array}
$$

\subsubsection{Compositional State Models}

\begin{proposition}[State Properties]\label{prop:ssm}
  A compositional concrete or symbolic state model, $\gstate =
\langle\sset{\gstate}, \ceval{~}{}, \kwT{ea} \rangle$,
satisfies the following properties:
\begin{enumerate}[leftmargin=*,label=\arabic*)]
\item $|\gstate| =   (|\gstate|, \stcomp,
  \stzero)$ is a PCM;

\item the state action execution function,  $\kwT{ea} :
  \sset{\gstate} \times (\mactions \uplus \sactions) \times \gval
  \pmap \power{\sset{\gstate} \times \gval \times \results}$,
  for
 $\fcolactionx \sigma \act \gv {} {\left\{ (\sigma_i,
     \gv_i)^{\result_i}|_{i \in I} \right\}}$, 
     satisfies properties 1-4 of Proposition~\ref{prop:esm}, as well as the following properties:
  \begin{enumerate}[left=0cm,label=\arabic*)]
  \setcounter{enumi}{4}
   \item non-missing executions of consumers, getters, and setters, as well as erroneous executions of producers and all executions of  actions in $\sactions$ are frame-preserving, with the latter given below formally, for $\act \in \sactions$: 
$$
\begin{array}{l}
\st.\act(\gv) \rightsquigarrow (\st', \gv')^\result \gand
(\st \stcomp \st_{f}) \text{ defined} \gand (\access{\Rnm(\st')}{ctx} \land \access{\st_f}{ctx}~\sat)
\\[0.5mm] \qquad 
\implies (\st \stcomp \st_{f}).\act(\Rnm(\gv)) \rightsquigarrow (\Rnm(\st') \stcomp \st_{f}, \Rnm(\gv'))^\result
\end{array}
$$
  where $\Rnm$ renames any symbols allocated by $\act$ that clash with those in $\sst_f$ to fresh symbols and keeps the other ones (including the fresh ones introduced by itself) the same.
  \end{enumerate}
\end{enumerate}
\end{proposition}

\begin{proof}
The proof that $(|\gstate|, \stcomp, \stzero)$ is a PCM is trivial, because of the PCM on the underlying memories and the commutativity and associativity of the allocator composition. The proof for frame preservation for memory actions comes from the frame preservation on memories, and the store and expression manipulating functions are automatically frame-preserving because they are unaffected by frame or the renaming. Note that in the case of single actions, $\rnm(\gv)$ equals $\gv$, but the more general approach is needed in order for the state action composition to work. We show the case for single uninterpreted symbol allocation, when the allocator of the frame contains the allocated symbol; the remaining cases are analogous:

\medskip
\case{\prooflab{uSym}}
\begin{proofpf}
  \step{s1}{
  \begin{pfenum}
  \item (H1) $\st.\kwT{uSym}(1) \rightsquigarrow (\st', \gv')^\result$
  \item (H2) $(\st \stcomp \st_{f}) \text{ defined}$ 
  \item (H3) $(\access{\Rnm(\st')}{ctx} \land \access{\st_f}{ctx}~\sat)$
  \end{pfenum}
  }
  \pflongnumbers
  \step{s2} {[by the hypotheses]
    \begin{pfenum}
  	  \item $\st = \tup{\cmem, \sto, \arec, \pc}$
	  \item $\alloc{\arec}{1}{\arec', [ \loc ]}{\clocs}$
  	  \item $\st' = \tup{\cmem, \sto, \arec', \pc}$
	  \item $\gv' = [ \loc ]$ 
	  \item $\result = \osucc$
    \end{pfenum} }
 \step{s3} {[by the hypotheses]
    \begin{pfenum}
      \item $\st_f = \tup{\cmem_f, \emptyset, \arec_f, \pc_f}$
  	  \item $\st \stcomp \st_f = \tup{\cmem \stcomp \cmem_f, \sto, \arec \stcomp \arec_f, \pc \land \pc_f}$
	  \item $\loc \in \arec(\clocs)$
	  \item $\rnm = \{ \loc \mapsto \loc' \}$, with $\loc' \notin (\arec \stcomp \arec_f)(\clocs)$
	  \end{pfenum}}
	  
 \step{s4} {$\alloc{(\arec \stcomp \arec_f)}{1}{\Rnm(\arec') \stcomp \arec_f, [ \loc' ]}{\clocs}$ [this is possible by non-determinism of allocation]}
 \step{s5} {$\rnm(\st') = \tup{\smem, \ssto, \rnm(\sarec'), \pc}$ [since $\loc$ is fresh w.r.t. $\st'$]}
 \step{s6} {$(\st \stcomp \st_{f}).\act(\gv) \rightsquigarrow (\Rnm(\st') \stcomp \st_{f}, \Rnm(\gv'))^\result$ [by 3.2, 4 and 5 and the definition of the $\kwT{uSym}$ action]}
\end{proofpf}
\end{proof}

We conclude by proving the following lemma:

\begin{lemma} 
\label{lem:fpfp}
The composition of two frame-preserving state actions is frame-preserving.
\end{lemma}

\begin{proof}
The proof is analogous to the proof of Proposition~\ref{prop:fpcomp}, with states instead of memories and the value passing specified by the two state compositions. The error cases for both are trivial. We consider the success case of $\acomp{\act_1}{\act_2}$, the other case is proven analogously:
\begin{proofpf}
     \pflongnumbers
 \step{1}{
      \assume{
        \begin{pfenum}
          \item (H1) $\faction{\st}{(\acomp{\act_1}{\act_2})}{\gv}{}{\st'}{\gv'}{}\result$
          \item (H2) $\st \stcomp \st_f \text{~defined}$ 
          \item (H3) $\access{\Rnm(\st')}{ctx} \land \access{\st_f}{ctx}~\sat$
        \end{pfenum}
      }
    }
      \step{2}{$ \exists \st'', \gv''$ [From Action Composition on H1]}
        \begin{pfenum*}
           \item $\faction{\st}{\act_1}{\gv}{}{\st''}{\gv''}{}\osucc$  \label{1m1}
           \item $\faction{\st''}{\act_2}{\gv''}{}{\st'}{\gv'}{}\result$ \label{1m2}
      \end{pfenum*}
       \step{3}{$\access{\Rnm(\st'')}{ctx} \land \access{\st_f}{ctx}~\sat$ [by H3 and monotonicity of contexts]}
      \step{4}{$\faction{(\st \stcomp \st_f)}{\act_1}{\rnm(\gv)}{}{\rnm(\st'') \stcomp \st_f}{\rnm(\gv'')}{}{\result}$ [by frame applied to 2.1, H2, and 3]}
            \step{3}{$\rnm(\st'') \stcomp \st_f \text{~defined}$ [by 4]}
            \step{4}{$\faction{\rnm(\st'')}{\act_2}{\rnm(\gv'')}{}{\rnm(\st')}{\rnm(\gv')}{}\result$ [by properties of renamings and 2.2]}

      \step{6}{$\faction{(\rnm(\st'') \stcomp \st_f)}{\act_2}{\rnm(\gv'')}{}{\rnm(\st') \stcomp \st_f}{\gv'}{}{\result}$ [by frame applied to 6, 5, and H3, noting idempotence of $\rnm$ as defined in the statement of the lemma]}
      \step{7}{$\faction{(\st \stcomp \st_f)}{(\acomp{\act_1}{\act_2})}{\gv}{} {\rnm(\st') \stcomp \st_f}{\gv'}{}\result$ [by 4 and 6]}
          \end{proofpf} 
\end{proof}

\subsubsection{Producers and Frame Preservation}
\label{sec:prodfp}

Producers, in general, are not frame preserving. However, they are only intended to be used to create new memory using symbols freshly allocated by the allocator, which are guaranteed to be fresh with respect to any frame, meaning that their specific use case is frame-preserving. To be able to include them in the frame preservation statement, however, the statement itself needs to be relaxed, as follows, for memory and state actions:
$$
{\begin{array}{l}
\faction \cmem {\act} \gv {\pc} {\cmem'} {\gv'} {\pc'} \result 
 \gand \sinv_{\!\pc \land \pc_f}(\cmem \stcomp \cmem_{f}) \gand (\pc' \land \pc_f~\sat)
\\ \qquad
 \implies 
\exists \pc''.~\faction {(\cmem \stcomp \cmem_f)} {\act} \gv {\pc \land \pc_f} {{\cmem'} \stcomp \cmem_f} {\gv'} {\pc' \land \pc'' \land \pc_f} \result \\[2mm]
\st.\act(\gv) \rightsquigarrow (\st', \gv')^\result \gand
(\st \stcomp \st_{f}) \text{ defined} \gand (\access{\Rnm(\st')}{ctx} \land \access{\st_f}{ctx}~\sat)
\\[0.5mm] \qquad 
\implies \exists \pc''.~(\st \stcomp \st_{f}).\act(\Rnm(\gv)) \rightsquigarrow (\Rnm(\st' \stcomp \pc'') \stcomp \st_{f}, \Rnm(\gv'))^\result
\end{array}}
$$
where $\pc''$ captures the separation constraints between the freshly allocated memory and the rest, and the notation $\st \stcomp \pc$ denotes $\st \stcomp \tup{\emptyset, \emptyset, \emptyset, \pc}$. Another issue is that memory and state action composition, defined like this, does not maintain frame preservation in general, because the $\pc''$ could take the execution away from the intended path. In particular, step 3 of the proofs of Lemma~\ref{prop:fpcomp} and Lemma~\ref{lem:fpfp} need not hold. However, again, in the specific case of producers, the introduced separation constraints cannot affect the path of the execution, since they do not introduce any connection between already present symbols, and the composition of two frame-preserving state actions remains frame-preserving. All of this means that the following requirement has to be met:

\begin{desiderata}[Frame Preservation: Producers in \gil Execution]
\label{des:prod}
All use cases of producers in the analysed \gil code must be frame-preserving and composition-preserving.
\end{desiderata}

\subsubsection{\gil Execution}

Given $\cf = \tup{\st, \cs, i}^o$, let $\cf \stcomp \st_{f}$ denote $\tup{\st \stcomp \st_f, \cs, i}^o$and let $\cf \stcomp \pc$ denote $\tup{\st \stcomp \tup{\emptyset, \emptyset, \emptyset, \pc}, \cs, i}^o$. We are now able to state the frame preservation property of the \gil execution:
\begin{theorem}[Frame Preservation: \gil Execution] 
\label{thm:exfp} Given a compositional concrete or symbolic state model, any corresponding non-missing \gil execution that satisfies Proposition~\ref{des:prod} is frame-preserving up to a renaming of allocated symbols:
$$
{\begin{array}{l}
\prog \vdash \cf \semarrow^n_C \cf' \gand
(\cf \stcomp \st_{f}) \text{ defined} \gand (\access{\access{\Rnm(\cf')}{st}}{ctx} \land \access{\st_f}{ctx}~\sat)
\\[0.5mm] \qquad 
\implies \exists \pc''.~\prog \vdash (\cf \stcomp \st_{f}) \semarrow^n_C (\Rnm(\cf' \stcomp \pc'') \stcomp \st_{f})
\end{array}}
$$
\end{theorem}

\begin{proof}
The single-step part is proven by case analysis on the transition of the \gil semantics, with the only cases affected by frame being those of the memory action, conditional goto, and symbol allocation. All three rules amount to a composition of frame-preserving state actions, which is proven to be frame-preserving by Lemma~\ref{lem:fpfp}.
The multi-step part (moving from $\semarrow$ to $\semarrow^n$) is proven by induction, similarly to the memory/state action composition.
\end{proof}

\section{\gillian Specification and Verification}
\label{subsec:gilasrt}

We introduce \gillian assertions, first for memories and then for states, and extend consumers and producers from core predicates to assertions. We define assertion satisfiability in a novel, parametric way. We introduce over-approximating function specifications, define their meaning, give a procedure for their verification and extend the \gil execution semantics with rules for specification use. We prove that the extended \gil semantics maintains forward soundness and that the use of function specifications is frame-preserving.

\subsection{Memory Assertions}
\label{subsec:memasrt}
\begin{definition}\label{def:funspec}
The core predicates $\corePreds \ni \corePred$ of a compositional memory model $\cmemory(\gval, \corePreds) = \tup{\sset{\cmemory}, \sinv, \cea_{\!\corePreds}}$ induce the following assertion language:
$$
\begin{array}{r@{~\defeq~}l}
\Psass, \Qsass \in \Psasses{} & \emp \mid \cPred{\sexp_i \cdot \sexp_o} \mid \astar{\Psass}{\Qsass} 
\end{array}
$$
where $\emp$ denotes the empty memory and $\ast$ is the standard separating conjunction. For simplicity, with the separating conjunction, we work modulo commutativity and associativity.
\end{definition}

We lift consumers and producers from core predicates to memory assertions, using substitutions $\Theta \ni \subst : \lxs \pmap \gval$, which map symbolic variables to values and are then lifted to symbolic expressions, $\Theta \ni \subst : \sexps \pmap \gval$, overloading notation. 
Substitutions can be thought of as more general versions of interpretations. 

\begin{definition}[Memory Assertion Consumers]
\label{def:memasscons}
The consumer and producer actions are lifted from core predicates to memory assertions, resulting in the assertion consumer and producer functions: 
$$
\fgetter{}, \fsetter{} : \sset{\cmemory} \times \Theta \times \Psasses{} \times \pcs \pmap \power{\sset{\cmemory} \times \gval \times \results \times \pcs}
$$
defined as follows (using the same pretty-printing as for other action execution functions):
\begin{mathpar}
\inferrule{\pc~\sat}{
\faction{\cmem}{\fgetter{\subst}}{\emp}\pc{\cmem}{\true}\pc{\osucc}}
\and
\inferrule
{\faction{\cmem}{\fgetter{\corePred}}{\subst(\sexp_i)}\pc{\cmem'}{\gv_o'}  {\pc'}{\osucc} \\\\ \pc'' = (\pc' \land \gv_o' = \subst(\sexp_o))) \and \pc''~\sat}
{\faction{\cmem}{\fgetter{\subst}}{\cPred{\sexp_i \cdot \sexp_o}}\pc{\cmem'}{\true}{\pc''}{\osucc}}
\and
\inferrule
{\faction{\cmem}{\fgetter{\corePred}}{\subst(\sexp_i)}\pc{\cmem'}{\gv_o'} {\pc'} {\osucc} \\\\ \pc'' = (\pc' \land \gv_o' \neq \subst(\sexp_o)) \and \pc''~\sat}
{\faction{\cmem}{\fgetter{\subst}}{\cPred{\sexp_i \cdot \sexp_o}}\pc{\cmem'}{\false}{\pc''}{\oerr}}
\and
\inferrule
{\faction{\cmem}{\fgetter{\corePred}}{\subst(\sexp_i)}\pc{\cmem}{\false}  {\pc'}{\result} \and \result \neq \osucc}
{\faction{\cmem}{\fgetter{\subst}}{\cPred{\sexp_i \cdot \sexp_o}}\pc{\cmem}{\false}{\pc'}{\result}}
\end{mathpar}

\begin{mathpar}
\inferrule
{
{\begin{array}{l}
\faction{\cmem}{\fgetter{\subst}}{\Psass/\Qsass}\pc{\cmem'}{\true}{\pc'}{\osucc} \\
\faction{\cmem'}{\fgetter{\subst}}{\Qsass/\Psass}{\pc'}{\cmem''}{\gv'}{\pc''}{\result}
\end{array}} \and 
\cmem''' = 
   {\begin{cases}
     \cmem'', & \text{if } \result = \osucc \\
     \cmem,   & \text{otherwise}
   \end{cases}}}
{\faction{\cmem}{\fgetter{\subst}}{\Psass \ast \Qsass}\pc{\cmem'''}{\gv'}{\pc''}{\result}}
\and
\inferrule
{\faction{\cmem}{\fgetter{\subst}}{\Psass/\Qsass}\pc{\cmem}{\false}{\pc'}{\result}}
{\faction{\cmem}{\fgetter{\subst}}{\Psass \ast \Qsass}\pc{\cmem}{\false}{\pc'}{\result}}
\end{mathpar}

\begin{mathpar}
\inferrule{\pc~\sat}{
\faction{\cmem}{\fsetter{\subst}}{\emp}\pc{\cmem}{\true}\pc{\osucc}}
\and
\inferrule
{\faction{\cmem}{\fsetter{\corePred}}{\subst(\sexp_i) \cdot \subst(\sexp_o)}\pc{\cmem'}{\gv}{\pc'}{\result}}
{\faction{\cmem}{\fsetter{\subst}}{\cPred{\sexp_i \cdot \sexp_o}}\pc{\cmem'}{\gv}{\pc'}{\result}}
\\
\inferrule
{
{\begin{array}{l}
\faction{\cmem}{\fsetter{\subst}}{\Psass/\Qsass}\pc{\cmem'}{\true}{\pc'}{\osucc} \\
\faction{\cmem'}{\fsetter{\subst}}{\Qsass/\Psass}{\pc'}{\cmem''}{\gv'}{\pc''}{\result}
\end{array}} \and 
\cmem''' = 
   {\begin{cases}
     \cmem'', & \text{if } \result = \osucc \\
     \cmem,   & \text{otherwise}
   \end{cases}}}
{\faction{\cmem}{\fsetter{\subst}}{\Psass \ast \Qsass}\pc{\cmem'''}{\gv'}{\pc''}{\result}}
\and
\inferrule
{\faction{\cmem}{\fsetter{\subst}}{\Psass/\Qsass}\pc{\cmem}{\false}{\pc'}{\oerr}}
{\faction{\cmem}{\fsetter{\subst}}{\Psass \ast \Qsass}\pc{\cmem}{\false}{\pc'}{\oerr}}
\end{mathpar}
\end{definition}

The following propositions can be derived straightforwardly from their core-predicate counterparts, by induction on the assertion consumed/produced:

\begin{propos}[Uniqueness of Memory Production] If $\faction{\cmem}{\fsetter{\subst}}{\Psass}\pc{\cmem'}{\true}{\pc'}\osucc$ and $\faction{\cmem}{\fsetter{\subst}}{\Psass}\pc{\cmem''} {\true}{\pc''}\osucc$, then $\cmem' = \cmem''$.
\end{propos}

Using this proposition, we can extend the notion of CPR to memory assertions, yielding the memory assertion resource function $\resfun : \Theta \tmap \Psasses{} \pmap \sset{\cmemory}$, as follows:
$$
\resfun_\subst(\Psass) = \cmem~~\Longleftarrow~~\faction{\stzero}{\fsetter{\subst}}{\Psass}-{\cmem}{\true}{-}\osucc
$$
To state the following properties, we also define a function $\resfun^i_\subst(\Psass)$, which takes an assertion and returns its core predicates together with their in-parameters, as follows:
$$
\begin{array}{r@{~\defeq~}l}
\resfun^i_\subst(\emp) & \emptyset \\
\resfun^i_\subst(\cPred{\gv_i \cdot \gv_o}) & \{ (\corePred, \gv_i) \}\\
\resfun^i_\subst(\Psass \ast q) & \resfun^i_\subst(\Psass) \cup \resfun^i_\subst(q)
\end{array}
$$
This means, in particular, that if $\resfun^i_\subst(\Psass) = \resfun^i_\subst(q)$, then $p$ and $q$ describe the same resource, possibly with different out-parameters.

\begin{propos}
[Assertion Consumers and Producers]
\label{prop:acpsi}
Assertion consumers and producers obey the following properties, which are analogous to those given in Definition~\ref{def:smm} and the accompanying derived properties:
$$
\begin{array}{r@{~}c@{~}l}
\faction{\cmem}{\fsetter{\subst}}{\Psass}\pc{\cmem'}{\true}{\pc'}{\osucc} & \implies & \cmem' = \cmem \stcomp \resfun_\subst(\Psass)
\\[1mm]
\faction{\cmem}{\fsetter{\subst}}{\Psass}\pc{\cmem'}{\gv'}{\pc'}{\result} & \implies & \result \neq \omiss
\\[1mm]
\faction{\cmem}{\fgetter{\subst}}{\Psass}\pc{\cmem'}{\true}{\pc'}{\osucc} & \implies & \pc' \Leftrightarrow \pc \land  (\cmem = \cmem' \stcomp \resfun_\subst(\Psass))
\\[1mm]
\sinv_{\!\pc}(\cmem) \gand \pc' \Leftrightarrow (\cmem = \cmem' \stcomp \resfun_\subst(\Psass)) & \implies & \faction \cmem {\fgetter{\corePred}} {\gv_i} {\pc} {\cmem'} {\true} {\pc \land \pc'} \osucc \\[1mm]
\faction{\cmem}{\fgetter{\subst}}{\Psass}\pc{\cmem'}{\true}{\pc'}{\osucc} \gand \resfun^i_\subst(\Psass) = \resfun^i_\subst(q) & \implies & \faction{\cmem'}{\fsetter{\subst}}{q} {\pc'} {\cmem' \stcomp \resfun_\subst(\Psass)} \true {\pc'} {\osucc}  \\[1mm]
\faction{\cmem}{\fsetter{\subst}}{\Psass}\pc{\cmem'}{\true}{\pc'}{\osucc} & \implies & \faction{\cmem'}{\fgetter{\subst}}{\Psass} {\pc'} \cmem {\true} {\pc'} {\osucc}
\end{array}
$$
\end{propos}

\myparagraph{Correctness: Forward Soundness and Backward Completeness}
The forward soundness and backward completeness requirements lift from core predicates to memory assertions. To state this, we extend the notion of interpretation to symbolic substitutions, as follows:
$$
\sintp(\ssubst) = \{ (\lx, \sintp(\sexp)) \mid (\lx, \sexp) \in \ssubst \}
$$
which is then lifted in the standard way to symbolic expressions, inducing that $(\sintp(\ssubst))(\sexp) = \sintp(\ssubst(\sexp))$.

\begin{propos}\label{prodcons:mass} 
Given a concrete and a symbolic compositional memory model, $\cmemory(\vals, \corePreds)$ and $\smemory(\sexps, \corePreds)$,
if their consumers and producers satisfy Requirements~\ref{des:memact}/Requirements~\ref{req:mem:fs}, then the lifted memory assertion consumers and producers also satisfy those requirements (stated below for~consumers):
\begin{mathpar}
\inferrule[Memory Assertion Consumers: Forward Soundness (MAC-FS)]{}{
          {\begin{array}{l}
\fcolactionx \smem {\fgetter{\ssubst}} \Psass {\pc} {\{ (\smem_i, \sexp_i)^{\result_i}_{\pc_i} |_{i \in I}\}}  \gand  \sintp(\pc) = \true \gand 
\action{\sint(\smem)}{\fgetter{\sintp(\ssubst)}}{\Psass}{(\cmem', \vl')^\result} \\ \qquad \implies \exists i.~(\cmem', \vl')^\result_\true = \sintp((\smem_i, \sexp_i)^{\result_i}_{\pc_i})
 \end{array}}}
 \and
\inferrule[Memory Assertion Consumers - Backward Completeness (MAC-BC)]
  {}{
  {\hspace*{-0.15cm}\begin{array}{l}
\faction{\smem}{\fgetter{\ssubst}}{\Psass}{\pc}{\smem'}{\sexp'}{\pc'}\result
  \land \sint(\pc') = \true \land \smem, \ssubst \in \domain(\sintp)
  \\ \qquad \implies
        \caction{\sintp(\smem)}{\fgetter{\sintp(\ssubst)}}{\Psass}{(\sintp(\smem'), \sintp({\sexp'}) )^\result}
        \end{array}}
 }
\end{mathpar}
\end{propos}

\begin{proof}
By induction on the structure of the memory assertion $\Psass$, noting that memory assertion consumers and producers have full coverage by definition.
We prove the base case of the consumers for both backward completeness and forward soundness. The case of $\emp$ is trivial, and the inductive cases follow straightforwardly. The producer cases are proven analogously.

\begin{proofpf}		
\case{Forward Soundness, Core Predicate Consumers:}
  \pflongnumbers
    \step{2ass}{
      \begin{pfenum}
        \item[(H1)] $\fcolactionx \smem {\fgetter{\ssubst}} \Psass {\pc} {\{ (\smem_i, \sexp_i)^{\result_i}_{\pc_i} |_{i \in I}\}}$
        \item[(H2)] $\sint(\pc) = \true$
        \item[(H3)] $\action{\sint(\smem)}{\fgetter{\sintp(\ssubst)}}{\Psass}{(\cmem', \vl')^\result}$
      \end{pfenum}
    }
		\step{2a}{Obtain $i$, such that $\sintp(\pc_i) = \true$ [by Proposition~\ref{prop:intermezzo}, H1 and H2]}
		\step{3}{$\faction{\smem}{\fgetter{\ssubst}}{\Psass}\pc{\smem_i}{\sexp_i}{\pc_i}{\result_i}$ [by H1 and 2]}
		\step{xb}{$\smem, \ssubst \in \domain(\sintp)$ [by H3]}
		\step{xd}{$\caction{\sintp(\smem)}{\fgetter{\sintp(\ssubst)}}{\Psass}{(\sintp(\smem'), \sintp({\sexp'}) )^{\result_i}}$ [by MAC-BC, using 3, 2, and 4]}
		\step{xe}{$\exists i.~(\cmem', \vl')^\result_\true = \sintp((\smem_i, \sexp_i)^{\result_i}_{\pc_i})$ [the witness is $i$, by H3, H5, and determinism of concrete execution]}
\end{proofpf}
\begin{proofpf}
\case{Backward Completeness, Core Predicate Consumers, $\osucc$ ($\oerr$/$\omiss$ is analogous):}
  \pflongnumbers
    \step{2ass}{
      \begin{pfenum}
        \item[(H1)] $\faction{\smem}{\fgetter{\ssubst}}{\cPred{\sexp_i \cdot \sexp_o}}{\pc}{\smem'}{\sexp'}{\pc'}\osucc$
        \item[(H2)] $\sint(\pc') = \true$
        \item[(H3)] $\smem, \ssubst \in \domain(\sintp)$
      \end{pfenum}
    }
		\step{2a}{[by H1 and H2]}
		\begin{pfenum*}
			\item $\faction{\smem}{\fgetter{\corePred}}{\ssubst(\sexp_i)}\pc{\smem'}{\sexp_o'}{\pc''}{\osucc}$ 
			\item $\pc' = (\pc'' \land \sexp_o' = \ssubst(\sexp_o)))$ 
			\item $\sintp(\pc'') = \true \gand \sintp(\sexp_o') = \sintp(\ssubst(\sexp_o))$
		\end{pfenum*}
		\step{xb}{$\ssubst(\sexp_i) \in \domain(\sintp)$ [by 2.1 and H3]}
		\step{xc}{$\faction{\sintp(\smem)}{\sintp(\fgetter{\corePred})}{\ssubst(\sexp_i)}{}{\sintp(\smem')}{\sintp(\ssubst(\sexp_o))}{}{\osucc}$ [by MA-BC of consumers, 2.1, 2.3, and 3]}
		\step{xd}{$\caction{\sintp(\smem)}{\fgetter{\sintp(\ssubst)}}{\Psass}{(\sintp(\smem'), \sintp({\sexp'}) )^\osucc$ [by definition, using 4 and H2]}}\end{proofpf}
\end{proof}

\myparagraph{Correctness: Frame Preservation}
We can easily prove, by induction on the structure of the assertion, that assertion consumption is frame-preserving (cf.~Proposition~\ref{prop:succexc}), and that consumption of an assertion $p$ followed by the production of an assertion $q$ such that $\resfun^i_\subst(\Psass) = \resfun^i_\subst(q)$ is frame-preserving in the sense of Definition~\ref{def:smm} (cf. Proposition~\ref{app:prop:comp}).
Moreover, this means that if $\resfun^i_\subst(q) \subseteq \resfun^i_\subst(p)$, that is, if $p$ has ``less'' resource than $q$, then the above composition is frame-preserving. We discuss the implications of $p$ having resource that is not in $q$ shortly, in the context of specification use.

\subsection{State Assertions}
\label{subsec:stasrt}
\begin{definition}\label{def:masrt}
A compositional state model $\gstate(\gval,
\mathsf{I}, \corePreds, \cmemory(\gval, \corePreds), \allocator(\gval)) \defeq
\langle\sset{\gstate}, \ceval{~}{}, \kwT{ea} \rangle$, with the underlying memory model $\cmemory(\gval, \corePreds)
\defeq \tup{\sset{\cmemory}, \sinv,$  $\cea}$, induces the following assertion language:
$$
\begin{array}{r@{~\defeq~}l}
   \Pass, \Qass \in \Passes{} & \{ ( 
   \Psass, \pc) \mid  
   \Psass \in \Psasses{}, \pc \in \pcs
   \}  
\end{array}
$$
where $\Psasses{}$ denotes the set of memory assertions induced by $\corePreds$. For convenience, $(\Psass, \pc)$ is onward denoted by $\Psass \land \pc$.
\end{definition}

\newcommand{\gmode}{m}
\newcommand{\gmodes}{\mathit{Modes}}
\newcommand{\omode}{o}
\newcommand{\umode}{u}

We next define state assertion consumption and production. Note that state assertions do not have explicit existential quantification; the treatment of existentials in specifications will be evident in the appropriate \gil rules.

\begin{definition}[State Assertion Consumers and Producers]
The state assertion consumer and producer functions,
$
\fgetter{}, \fsetter{} : \sset{\gstate} \times 
\Theta \times \Passes{} \pmap \power{\sset{\gstate} \times \gval \times \results}
$,
are defined as follows:
\begin{mathpar}
\inferrule
{
 \st = \tup{\cmem', \sto, \arec, \pc'} \\\\
 \faction {\cmem'} {\fgetter{\subst}} \Psass {\pc'} {\cmem''}{\true} {\pc''} {\osucc} \\\\
	\pc'' \land \subst(\pc)~\sat
	  \\\\
 \st' = \tup{\cmem'', \sto, \arec, \pc'' \land \subst(\pc)}
}
{\faction{\st}{\fgetter{\subst}}{\Psass \land \pc}{}{\st'} \true {} \osucc}
\and
\inferrule
{
 \st = \tup{\cmem', \sto, \arec, \pc'} \\\\
 \faction {\cmem'} {\fgetter{\subst}} \Psass {\pc'} {\cmem''}{\true} {\pc''} {\osucc} \\\\
	\pc'' \land \lnot\subst(\pc)~\sat\\\\
 \st' = \tup{\cmem, \sto, \arec, \pc'' \land \lnot \subst(\pc)}
}
{\faction{\st}{\fgetter{\subst}}{\Psass \land \pc}{}{\st'} \false {} \oerr}
\\
\inferrule
{
 \faction {\cmem'} {\fgetter{\subst}} \Psass {\pc'} {\cmem'}{\false} {\pc''} {\result} \and \result \neq \osucc
}
{\faction{\tup{\cmem', \sto, \arec, \pc'}}{\fgetter{\subst}}{\Psass \land \pc}{}{\tup{\cmem', \sto, \arec, \pc''}} \false {} \result}
\\
\inferrule
{
 \st = \tup{\cmem', \sto, \arec, \pc'} \\\\
 \faction {\cmem'} {\fsetter{\subst}} \Psass {\pc'} {\cmem''}{\true} {\pc''} {\osucc} \\\\
  \pc'' \land \subst(\pc)~\sat \\\\
 \st' = \tup{\cmem'', \sto, \arec, \pc'' \land \subst(\pc)}
}
{\faction{\st}{\fsetter{\subst}}{\Psass \land \pc}{}{\st'} \true {} \osucc}
\and
\inferrule
{
 \st = \tup{\cmem', \sto, \arec, \pc'} \\\\
 \faction {\cmem'} {\fsetter{\subst}} \Psass {\pc'} {\cmem'}{\false} {\pc''} {\osucc} \\\\
 \result \neq \osucc \\\\
 \st' = \tup{\cmem'', \sto, \arec, \pc''}
}
{\faction{\st}{\fsetter{\subst}}{\Psass \land \pc}{}{\st'} \false {} \result}
\end{mathpar}
\end{definition}

\myparagraph{Correctness}
We can easily show that forward soundness, backward completeness, and frame preservation lift from memory to state assertions. However, note that producers no longer maintain full coverage in general, because there might not exist a resulting branch with context $\pc'' \land \lnot\subst(\pc)$. We discuss the impact of this further in the context of procedure specifications.

\subsection{Assertion Satisfiability} 
\label{subsec:asrtsat}
We now give our definition of assertion satisfiability, stated in terms of memory assertion resource:

\begin{definition}[Satisfiability]
  \label{sat}
The \emph{satisfiability relation}, stating that memory
$\cmem'$ and context $\pc'$ satisfy assertion $\Psass \land \pc$ under
substitution~$\subst$,  is defined by:
$$
\begin{array}{r@{~}c@{~}l@{~}c@{~}l}
\cmem', \pc', \subst & \models & \Psass \land \pc & \iff & \sinv_{\!\pc'}(\cmem') \gand \pc' \vdash \cmem' = \resfun_{\subst}(\Psass)  \gand \pc' \vdash \subst(\pc)
\end{array}
$$
and is lifted to states as: $\tup{\cmem', -, -, \pc'}, \subst \models \Psass \land \pc$
if and only if $ \cmem', \pc', \subst \models \Psass \land \pc$.
\end{definition}

The following lemma gives further insight into the connection between consumers, producers, and satisfiability.

\begin{lemma}[Consumers, Producers, and Satisfiability]
\label{lem:consprodsat}
$$
\begin{array}{r@{~}c@{~}l}
\cmem', \pc', \subst \models \Psass & \implies & \faction{\cmem'}{\fgetter{\subst}}{\Psass}{\pc'}{\stzero}{\true}{\pc'}\osucc \\[1mm]
\action{\st}{\fgetter{\subst}}{P}{(\st', \true)^\osucc} & \implies & \exists \st''.~(\access{\st_j}{\ctx} \vdash \st = \st_j \stcomp \st'') \gand (\st'', \subst \models P) \\[1mm]
\action{\st}{\fsetter{\subst}}{P}{(\st', \true)^\osucc} 
&\implies &
\exists \st''.~\st' = \st \stcomp \st'' \gand \st'', \subst \models P
\end{array}
$$
\end{lemma}
\begin{proof}
We prove the second part; the first part is trivial and the second part is proven analogously.
\begin{proofpf}
     \pflongnumbers
 \step{1}{
      \assume{(H1) $\action{\st}{\fgetter{\subst}}{P}{(\st', \true)^\osucc}$}
    }
      \step{2}{[From H1, by definition]}
        \begin{pfenum*}
 \item $\st = \tup{\cmem', \sto, \arec, \pc'}$
 \item $\faction {\cmem'} {\fgetter{\subst}} \Psass {\pc'} {\cmem_j}{\true} {\pc_j} {\osucc}$ 
 \item $\pc_j \land \subst(\pc)~\sat$
\item $\st_j = \tup{\cmem_j, \sto, \arec, \pc_j \land \subst(\pc)}$
      \end{pfenum*}
          \step{5}{[by Proposition~\ref{prop:acpsi} and 2.2]}
        \begin{pfenum*}
 \item $\pc_j \Leftrightarrow (\pc \land (\cmem' = \cmem_j \stcomp \resfun_\subst(\Psass)))$
      \end{pfenum*}
      \step{7}{Take $\st'' = \tup{\resfun_\subst(\Psass), \emptyset, \arec, \pc_j \land \subst(\pc)}$, where $\arec$ contains the symbols in $\resfun_\subst(\Psass)$ and $\pc_j$}
      \step{8}{$\access{\st_j}{\ctx} \vdash \st = \st_j \stcomp \st''$ [by definition of composition, 2.1, 2.4, 3.1, and 4]}
      \step{8b}{$\faction {\cmem_j} {\fsetter{\subst}} \Psass {\pc_j} {\cmem_j \stcomp \resfun_\subst(\Psass)}{\true} {\pc_j} {\osucc}$ (by Proposition~\ref{prop:acpsi} and 2.2)}
      \step{8c}{$\faction {\stzero} {\fsetter{\subst}} \Psass {\pc_j} {\resfun_\subst(\Psass)}{\true} {\pc_j} {\osucc}$ (by frame cancellation of successful producers)}
      \step{8d}{$\faction {\resfun_\subst(\Psass)} {\fsetter{\subst}} \Psass {\pc_j} {\stzero}{\true} {\pc_j} {\osucc}$ (by Proposition~\ref{prop:acpsi})}
      \step{9}{$\st'', \subst \models \Psass \land \pc$ [by 8 and 3.4]}
          \end{proofpf}  
\end{proof}

\subsection{Procedure Specifications}
\label{subsec:procspecs}
\gillian procedure 
 specifications have the form $\{\lx, P\} f(x)
\{Q\}^{\sexp}$, where $f$ is the procedure identifier, $x$ is the
procedure parameter, $\lx$ is the symbolic variable holding the value of
$x$, $\Pass$ is the pre-condition, $\Qass$ is the post-condition,
and $\sexp$ is the return value of the procedure,
with the following constraints:

\begin{itemize}
\item the pre-condition does not have explicit existential quantification;
\item program variables do not appear in the pre- or the
  post-condition, and the procedure parameter $x$ is accessed using symbolic variable $\lx$; and
\item symbolic variables that appear in a pre-condition are implicitly universally quantified, and can be re-used in the corresponding post-condition
\item symbolic variables that appear in a pre-condition are implicitly existentially quantified
\item any resource in the post-condition created within the function is existentially quantified.
\end{itemize}

We extend \gil programs with procedure specifications, accessible
via $\prog.\kwT{specs}$, and the \gil semantics with rules for re-using specifications, given in Figure~\ref{fig:specrules}. The rules are simplified, in that they assume to have the substitution upfront in the procedure call; in the implementation, there exists a unification algorithm that, starting from the procedure parameter and using the consumers, learns the substitution. Additionally, for simplicity, we allow one specification per function.

\begin{figure}[!t]
\begin{mathpar}
\small
\inferrule[\textsc{Call - Spec Re-Use - Success}]
  {
  {\begin{array}{lr}
     \cmd(\prog, \cs, i) = y := \e(\e') \text{ with } \subst & \text{\footnotesize procedure call with substitution $\subst$}
      \\
      \eex{\st}{\e} \rightsquigarrow f \quad \eex{\st}{\e'} \rightsquigarrow \gv' & \text{\footnotesize get procedure id and parameter value}
      \\
      \{\lx, P \} f(x) \{Q\}^{\sexp} \in {\prog}.\kwT{specs} & \text{\footnotesize get procedure specification} \\
      \subst' = \subst[\lx \mapsto \gv'] & \text{\footnotesize extend  substitution with parameter value}
      \\
      \fcolactionx{\st}{\fgetter{\subst'}}{P}{}{\{ (\st_j, \true)^\osucc} |_{j \in J} \} & \text{\footnotesize consume pre-condition} \\
      j \in J & \text{\footnotesize select a branch} \\
      \lstxs = \textsf{svars}(Q) \setminus \textsf{svars}(P) & \text{\footnotesize get the existentials of the post} \\
      \action{\st_j}{\kwT{iSym}}{|\lstxs|}{(\st_j', \lstvs)^\osucc} & \text{\footnotesize generate post-condition existentials} \\
      \subst_{j}' = \subst'[ \lstxs \mapsto \lstvs ] & \text{\footnotesize extend substitution with generated existentials} \\
      \action{\st_j'}{\fsetter{\subst_j'}}{Q}{(\st_j'', \true)^\osucc} & \text{\footnotesize produce post-condition}
      \\
      \action{\st_j''}{\kwT{setVar}_y}{\subst_j'(\sexp)}{\st'} & \text{\footnotesize assign return value}
   \end{array}}
  }{
    \semtrans{\st, \cs, i}{\st', \cs, i{+}1}{}{\prog}{}
  }
\and
\inferrule[\textsc{Call - Spec Re-Use - Error}]
  {
  {\begin{array}{lr}
     \cmd(\prog, \cs, i) = y := \e(\e') \text{ with } \subst & \text{\footnotesize procedure call with substitution $\subst$}
      \\
      \eex{\st}{\e} \rightsquigarrow f \quad \eex{\st}{\e'} \rightsquigarrow \gv' & \text{\footnotesize get procedure id and parameter value}
      \\
      \{\lx, P \} f(x) \{Q\}^{\sexp} \in {\prog}.\kwT{specs} & \text{\footnotesize get procedure specification} \\
      \subst' = \subst[\lx \mapsto \gv'] & \text{\footnotesize extend  substitution with parameter value}
      \\
            \fcolactionx{\st}{\fgetter{\subst'}}{P}{}{\{ (\st_j, \true)^{\result_j}}|_{j \in J} \} & \text{\footnotesize consume pre-condition} \\
      \exists j \in J.~\result_j \neq \osucc & \text{\footnotesize consumption is not fully successful} \\
      \st_j' = \st \stcomp \access{\st_j}{\ctx} & \text{\footnotesize construct final state} \\
      o_j = (\text{if $\result_j = \omiss$ then $\mathsf{M}$ else $\mathsf{E}$)}& \text{\footnotesize construct final outcome, none successful} \\
      \gv'_j = (\text{if $\result_j = \osucc$ then $\false$ else $\gv_j$)}& \text{\footnotesize construct return value}
   \end{array}}
  }{
    \semtrans{\st, \cs, i}{\st_j', \cs, i}{}{\prog}{o_j(\gv_j')}
  }
\end{mathpar}
\vspace*{-0.3cm}
\caption{\gil Execution: Specification Use}
\vspace*{-0.3cm}
\label{fig:specrules}
\end{figure}

Lemma~\ref{lem:consprodsat} allows us to give these two rules a standard separation logic interpretation. In particular, the consumption of the pre-condition $P$ attempts to discover and then remove the part of the state $\st$ that satisfies $P$. If there exists a branch for which this is not possible, the procedure cannot be called with the given specification and an error is returned. This is required because the specifications used here are verification specifications, and verification has to cover all paths. Also note that, if $\fcolactionx{\st}{\fgetter{\subst}}{p \land \pc}{}{\{ (\st_j, \true)^\osucc} |_{j \in J} \}$ holds, that means that $\access{\st}{ctx} \vdash \subst(\pc)$, which is the standard definition of assertion satisfiability for pure assertions in over-approximating analyses.
Otherwise, we have that the resulting state $\st_j$ effectively corresponds to the frame of the function call, and that the subsequent production extends that frame with a part of the state that satisfies the post-condition~$Q$. Observe that, due to the way pure assertion consumption is handled, that the context never gets weakened, even if the pure part of the post-condition of the procedure is weaker than that of the pre. Finally, in the error cases, we do not allow the non-deterministic execution to proceed, and all branches are, therefore, converted to error branches.  

\myparagraph{The Verification Process in Gillian}
We discuss how \gillian verifies procedure specification using the \gil symbolic execution.

\newcommand{\lstes}{\hat{\mathit{es}}}

\begin{definition}[\gillian Verification Procedure]
A specification $\{\lx, P\} f(x) \{Q\}^{\sexp}$ is \emph{verified} by the \gil symbolic execution, denoted by $\models \{\lx, P\} f(x) \{Q\}^{\sexp}$, iff, taking $\ssubst = \mathsf{id}(\mathsf{svars}(P))$, $\ssto = \{ x \mapsto \ssubst(\lx) \}$, $\sarec$ such that $\sarec(\lxs) = \mathsf{svars}(P)$ and $\sarec(\clocs) = \emptyset$, $\sst = \tup{\emptyset, \ssto, \sarec, \true}$, and $\lstxs = \mathsf{svars}(Q) \setminus \mathsf{svars}(P)$, the following holds:
$$
\begin{array}{l}
\action{\sst}{\fsetter{\ssubst}}{P}{(\sst'', \true)^\osucc} \implies \\
\quad \prog \vdash \{ \tup{\sst'', \tup{f}, 0} \} \csemarrowfin \ssscf'~\gand \\ 
\qquad (\forall \tup{\sst', -, -}^{o(\sexp')} \in {\ssscf}'.~o = \mathsf{S} \gand (\exists \lstes.~\sst', \ssubst[\lstxs \mapsto \lstes] \models Q \gand  \access{\sst'}{ctx} \vdash \sexp' = (\ssubst[\lstxs \mapsto \lstes])(\sexp)))
\end{array}
$$
\end{definition}

To verify a specification in \gillian, we start by constructing an initial state $\sst$ that sets up the store for the execution of the procedure $f$ and has an  allocator that has already allocated the symbolic variables of the pre-condition, and then producing the pre-condition $P$ in this state. If this production fails, then the specification vacuously holds. Otherwise, we obtain the state $\st''$ that, by Lemma~\ref{lem:consprodsat}, satisfies $P$. Then, in order for the specification to be verified, the symbolic execution of $f$ starting from $\st''$ has to terminate, yielding a set of final configurations $\ssscf'$. Moreover, each final configuration has to have terminated successfully, each final state has to satisfy the post-condition, and each return value has to equal the desired one. Note that the substitution may need to be extended to accommodate the existentials of the post-condition. 

\myparagraph{Correctness: Forward Soundness}
In order to prove the following correctness statement, we extend interpretation to programs as well, due to having to interpret the co-domain of the substitutions $\theta$ provided in the procedure calls. This is needed because of the above-mentioned simplification, which elides the unification algorithm.

Observe that in order to preserve full coverage, any given specification $\{\lx, p \land \pc_p \} f(x) \{ q \land \pc_q \}^{\sexp}$ has to maintain that $\pc_p \Rightarrow \pc_q$, meaning that the post-condition cannot introduce new pure facts about existing symbolic variables, that is, cut execution paths. This is a standard requirement which can easily be met for bounded code, but has to be handled via user-defined predicates for code with unbounded loops or recursive calls. While Gillian does support user-defined predicates, the related machinery will be introduced in subsequent versions of this technical report.

Finally, note that backward completeness no longer holds in general because it is possible to define over-approximating user-defined predicates that can lead to loss of information due to predicate folding. If the specifications are not over-approximating, either due to functions not having unbounded executions or the user-defined predicates not losing information, then backward completeness does hold and the \gil execution still does not have false positive bug reports.

\begin{lemma}[Forward Soundness: \gil One-Step]\label{lem:gilfspecone}
Given a concrete and a symbolic compositional state model, $\gstate(\vals,
\vals, \corePreds, \cmemory(\vals, \corePreds), \allocator(\vals)) =
\langle\sset{\gstate}, \ceval{~}{}, \kwT{ea} \rangle$ and $\gstate(\sexps,
\lxs, \corePreds, \cmemory(\sexps, \corePreds),$ $\allocator(\sexps)) = \langle\sset{\hat\gstate}, \ceval{~}{}, \hat{\kwT{ea}} \rangle$, respectively, 
if the consumers and producers in  satisfy Requirements~\ref{req:mem:fs}, and assuming that procedure specifications do not strengthen the context, then the corresponding \gil one-step symbolic semantics satisfies the following property:
\begin{mathpar}
\inferrule[GIL: Forward Soundness: Interpretations (GIL-FS-CFG-I)]{}{
          {\begin{array}{l}
\prog \vdash \scf \rightarrow \ssscf'  \gand  
\sintp(\prog) \vdash  \sintp(\scf)  \semarrow {\cf'}  \implies \exists \varepsilon' \geq \varepsilon, \scf' \in \ssscf'.~\cf' = \sintpe{\varepsilon'}(\scf')
 \end{array}}}
\end{mathpar}
\end{lemma}
\begin{proof}
Given Proposition~\ref{cor:spec:simpl}, we only need to prove the desired property for the new \prooflab{Call - Spec Re-Use} rule. Note that since this rule uses the collecting semantics in its premise, only one of its two variants can be applied at any given time. We prove the \prooflab{Call - Spec Re-Use - Success} case; the error case is analogous.
\case{\prooflab{Call - Spec Re-Use - Success}}
\begin{proofpf}
  \pflongnumbers
  \step{s1} {Hypotheses:
    \begin{pfenum}
      \item[(H1)] $\prog \vdash \scf \rightarrow \ssscf'$
      \item[(H2)] $\sintp(\prog) \vdash  \sintp(\scf)  \semarrow {\cf'}$ 
    \end{pfenum}}
  \step{s2} {[by H1, \prooflab{Call} case]
    \begin{pfenum}
     \item $\cmd(\prog, \hat\cs, i) = y := \e(\e') \text{ with } \subst$
     \item $\eex{\sst}{\e} \rightsquigarrow f \quad \eex{\sst}{\e'} \rightsquigarrow \sexp'$
     \item $\{\lx, P \} f(x) \{Q\}^{\sexp} \in {\prog}.\kwT{specs}$
     \item $\ssubst' = \ssubst[\lx \mapsto \sexp']$
     \item $\fcolactionx{\sst}{\fgetter{\ssubst'}}{P}{}{\{ (\sst_j, \true)^\osucc} |_{j \in J} \}$
     \item $\lstxs = \textsf{svars}(Q) \setminus \textsf{svars}(P)$
     \item $(\action{\sst_j}{\kwT{iSym}}{|\lstxs|}{(\sst_j', \lstxs')^\osucc})|_{j \in J}$ 
     \item $(\ssubst_{j}' = \ssubst'[ \lstxs \mapsto \lstxs'])|_{j \in J}$
     \item $(\fcolactionx{\sst_j'}{\fsetter{\ssubst_j'}}{Q}{}{\{ (\sst_{lj}, \true)^\osucc \mid_{lj \in L_j} \} })|_{j \in J}$
     \item $(\action{\sst_{lj}}{\kwT{setVar}_y}{\subst_j'(\sexp)}{\sst_{lj}'})|_{j \in J, lj \in L_j}$ [note that this transition is deterministic]
     \item $\ssscf' = \bigcup_{j \in J} \bigcup_{lj \in L_j} \{ \tup{\sst_{lj}, \hat\cs, i{+}1} \}$
    \end{pfenum}}
  \step{s3} {[by H2, \prooflab{Call - Spec Re-Use - Success} case]
    \begin{pfenum}
     \item $\cmd(\sintp(\prog), \cs, i) = y := \e(\e') \text{ with } \sintp(\subst)$
     \item $\eex{\sintp(\sst)}{\e} \rightsquigarrow f \quad \eex{\sintp(\sst)}{\e'} \rightsquigarrow \sintp(\sexp')$ [by forward soundness and determinism of state actions, noting that no allocation happens, and so $\varepsilon$ can stay the same]
     \item $\{\lx, P \} f(x) \{Q\}^{\sexp} \in {\sintp(\prog)}.\kwT{specs}$
     \item $\subst' = \sintp(\ssubst)[\lx \mapsto \sintp(\sexp')] = \sintp(\ssubst')$ [by 2.4]
     \item $\fcolactionx{\sintp(\sst)}{\fgetter{\sintp(\ssubst')}}{P}{}{\{ (\st_k, \true)^\osucc} |_{k \in K} \}$ [given 2.5, by forward soundness of state assertion consumers all concrete transitions must be successful, noting that no allocation happens, and so $\varepsilon$ can stay the same]
     \item $k \in K$
     \begin{pfenum}
     \item $\exists j \in J.~\st_k = \sintp(\sst_j)$ [by 2.5, 3.5, and forward soundness of state assertion consumers]
     \end{pfenum}
     \item $\lstxs = \textsf{svars}(Q) \setminus \textsf{svars}(P)$
     \item $\action{\sintp(\st_j)}{\kwT{iSym}}{|\lstxs|}{(\st_j', \lstvs)^\osucc}$ [note that this transition is not deterministic]
     \begin{pfenum}
	     \item $\exists \varepsilon' \geq \varepsilon.~\st_j' = \sintpe{\varepsilon'}(\sst_j') \land \lstvs = \sintpe{\varepsilon'}(\lstxs')$ [by forward soundness of $\kwT{iSym}$, 2.7, and 3.8]
	     \item $\varepsilon' = \varepsilon[\lstxs' \mapsto \lstvs]$ [by 3.8a]
     \end{pfenum}
     \item $\subst_{j}' = \subst'[ \lstxs \mapsto \lstvs ] = \sintpe{\varepsilon'}(\ssubst_j')$ [by 2.8]
     \item $\action{\sintp(\sst_j')}{\fsetter{\sintp(\ssubst_j')}}{Q}{(\sst_j'', \true)^\osucc}$
          \begin{pfenum}
	     \item $\exists lj \in L_j.~\st_j'' = \sintpe{\varepsilon'}(\sst_{lj})$ [by forward soundness of state assertion producers, 2.9, and 3.10, noting that no allocation happens, and so $\varepsilon'$ can stay the same]
     \end{pfenum}
     \item $\action{\sintpe{\varepsilon'}(\sst_{lj})}{\kwT{setVar}_y}{\sintp(\ssubst_j'(\sexp))}{\sst'}$
               \begin{pfenum}
	     \item $\sst' = \sintpe{\varepsilon'}(\sst_{lj}')$ [by forward soundness of $\kwT{setVar}$, 2.10, and 3.11, noting that no allocation happens, and so $\varepsilon'$ can stay the same]
     \end{pfenum}
    \end{pfenum}}
    
   \step{4}{$\exists \varepsilon' \geq \varepsilon, \scf' \in \ssscf'.~\cf' = \sintpe{\varepsilon'}(\scf')$ [for $\varepsilon' = \varepsilon'$ and $\scf' = \tup{\sst_{lj}', \hat\cs, i{+}1}$, given 2.11 and the fact that $\sintpe{\varepsilon'}(\cs) = \sintp(\cs)$ and therefore $\cf' = \tup{\sintpe{\varepsilon'}(\sst_{lj}'), \sintpe{\varepsilon'}(\hat\cs), i{+}1}$]}
\end{proofpf}

\end{proof}

From this Lemma, we can derive the following verification theorem, which has the same formulation as Theorem~\ref{lem:gttcsfs:mod} for partial verification, but given the ability to use procedure specifications, can handle code with recursive procedure calls and unbounded looping (which can be transformed to recursive procedure calls) and therefore captures full verification.

\begin{theorem}[\gil Verification] 
\label{lem:verif}
\gil execution satisfies the following property:
\begin{mathpar}
 \inferrule[GIL Verification]{}{
          {\begin{array}{l}
\prog \vdash \{ \scf \} \csemarrowfin \ssscf' \gand \cf \in \modls({\scf}) \gand~\prog \vdash  \cf  \csemarrowfin_C {\cf'} \\ \qquad \implies \exists \scf'.~\scf' \in \ssscf' \gand \cf' \in \modls(\scf') \gand \prog \vdash \scf \csemarrowfin_C \scf'
 \end{array}}}
\end{mathpar}
\end{theorem}

\myparagraph{Correctness: Frame Preservation}
To prove that Theorem~\ref{thm:exfp} extends to the \gil execution with specification use, we only need to prove frame preservation of the two rules introduced in Figure~\ref{fig:specrules}. The error case is a composition of expression evaluation and an assertion consumer, which are all frame-preserving, and is, therefore, frame-preserving itself. The successful case is a composition of expression evaluation, followed by a consumer for the pre-condition, generation of post-condition existentials, and the production of the post-condition. The critical point is the production of the post-condition, which may have additional resource with respect to the pre-condition. However, this resource has to be existentially quantified, which means that fresh symbols will be created to represent it. This, in turn, means that the production of the post-condition, much like the production of allocated memory, will be frame-preserving.
 
Finally, given the frame preservation properties of the \gil execution, we can prove compositionality of function specifications, stated as follows:
\begin{itemize}[leftmargin=*]
\item if $\models\! \{\lx, P\} f(x) \{Q\}^{o(\sexp)}$ and $P * R$ has a successful production, then $\models_{} \{\lx, P \ast R \} f(x) \{Q \ast R \}^{o(\sexp)}$.
\end{itemize}

%% file: main.bbl

\begin{thebibliography}{4}


\ifx \showCODEN    \undefined \def \showCODEN     #1{\unskip}     \fi
\ifx \showDOI      \undefined \def \showDOI       #1{#1}\fi
\ifx \showISBNx    \undefined \def \showISBNx     #1{\unskip}     \fi
\ifx \showISBNxiii \undefined \def \showISBNxiii  #1{\unskip}     \fi
\ifx \showISSN     \undefined \def \showISSN      #1{\unskip}     \fi
\ifx \showLCCN     \undefined \def \showLCCN      #1{\unskip}     \fi
\ifx \shownote     \undefined \def \shownote      #1{#1}          \fi
\ifx \showarticletitle \undefined \def \showarticletitle #1{#1}   \fi
\ifx \showURL      \undefined \def \showURL       {\relax}        \fi
\providecommand\bibfield[2]{#2}
\providecommand\bibinfo[2]{#2}
\providecommand\natexlab[1]{#1}
\providecommand\showeprint[2][]{arXiv:#2}

\bibitem[\protect\citeauthoryear{Maksimovi\'c, Ayoun, Santos, and
  Gardner}{Maksimovi\'c et~al\mbox{.}}{2021}]%
        {gilliancav}
\bibfield{author}{\bibinfo{person}{Petar Maksimovi\'c},
  \bibinfo{person}{Sacha{-}\'Elie Ayoun}, \bibinfo{person}{Jos\'e~Fragoso
  Santos}, {and} \bibinfo{person}{Philippa Gardner}.}
  \bibinfo{year}{2021}\natexlab{}.
\newblock \showarticletitle{{Gillian, Part II: Real-World Verification for
  JavaScript and C}}. In \bibinfo{booktitle}{\emph{Computer Aided Verification
  - 33rd International Conference, {CAV} 2020, Los Angeles, CA, USA, July
  18-23, 2021}}. \bibinfo{publisher}{Springer}.
\newblock
\newblock
\shownote{Accepted for publication.}


\bibitem[\protect\citeauthoryear{O'Hearn}{O'Hearn}{2020}]%
        {incohearn}
\bibfield{author}{\bibinfo{person}{Peter~W. O'Hearn}.}
  \bibinfo{year}{2020}\natexlab{}.
\newblock \showarticletitle{{Incorrectness Logic}}.
\newblock \bibinfo{journal}{\emph{Proc. {ACM} Program. Lang.}}
  \bibinfo{volume}{4}, \bibinfo{number}{{POPL}} (\bibinfo{year}{2020}),
  \bibinfo{pages}{10:1--10:32}.
\newblock
\urldef\tempurl%
\url{https://doi.org/10.1145/3371078}
\showDOI{\tempurl}


\bibitem[\protect\citeauthoryear{Raad, Berdine, Dang, Dreyer, O'Hearn, and
  Villard}{Raad et~al\mbox{.}}{2020}]%
        {isl}
\bibfield{author}{\bibinfo{person}{Azalea Raad}, \bibinfo{person}{Josh
  Berdine}, \bibinfo{person}{Hoang{-}Hai Dang}, \bibinfo{person}{Derek Dreyer},
  \bibinfo{person}{Peter~W. O'Hearn}, {and} \bibinfo{person}{Jules Villard}.}
  \bibinfo{year}{2020}\natexlab{}.
\newblock \showarticletitle{{Local Reasoning About the Presence of Bugs:
  Incorrectness Separation Logic}}. In \bibinfo{booktitle}{\emph{Computer Aided
  Verification - 32nd International Conference, {CAV} 2020, Los Angeles, CA,
  USA, July 21-24, 2020, Proceedings, Part {II}}}
  \emph{(\bibinfo{series}{Lecture Notes in Computer Science},
  Vol.~\bibinfo{volume}{12225})},
  \bibfield{editor}{\bibinfo{person}{Shuvendu~K. Lahiri} {and}
  \bibinfo{person}{Chao Wang}} (Eds.). \bibinfo{publisher}{Springer},
  \bibinfo{pages}{225--252}.
\newblock
\urldef\tempurl%
\url{https://doi.org/10.1007/978-3-030-53291-8\_14}
\showDOI{\tempurl}


\bibitem[\protect\citeauthoryear{Santos, Maksimovic, Ayoun, and Gardner}{Santos
  et~al\mbox{.}}{2020}]%
        {gillianpldi}
\bibfield{author}{\bibinfo{person}{Jos{\'{e}}~Fragoso Santos},
  \bibinfo{person}{Petar Maksimovic}, \bibinfo{person}{Sacha{-}{\'{E}}lie
  Ayoun}, {and} \bibinfo{person}{Philippa Gardner}.}
  \bibinfo{year}{2020}\natexlab{}.
\newblock \showarticletitle{{Gillian, Part I: a Multi-language Platform for
  Symbolic Execution}}. In \bibinfo{booktitle}{\emph{Proceedings of the 41st
  {ACM} {SIGPLAN} International Conference on Programming Language Design and
  Implementation, {PLDI} 2020, London, UK, June 15-20, 2020}},
  \bibfield{editor}{\bibinfo{person}{Alastair~F. Donaldson} {and}
  \bibinfo{person}{Emina Torlak}} (Eds.). \bibinfo{publisher}{{ACM}},
  \bibinfo{pages}{927--942}.
\newblock
\urldef\tempurl%
\url{https://doi.org/10.1145/3385412.3386014}
\showDOI{\tempurl}


\end{thebibliography}
